\documentclass[12pt,a4paper]{article}
\usepackage{amsmath,amssymb,amsthm} 
\usepackage[dvipdfmx]{graphicx}
\numberwithin{equation}{section}
\newtheorem{lemma}{Lemma}
\newtheorem{theorem}{Theorem}

\theoremstyle{remark}
\newtheorem{remark}{Remark}
\newcommand{\beq}{\begin{equation}}
\newcommand{\eeq}{\end{equation}}
\newcommand{\beqnn}{\begin{equation*}}
\newcommand{\eeqnn}{\end{equation*}}
\newcommand{\rd}{\partial}

\newcommand{\tp}[1]{\,{\vphantom{#1}}^\mathrm{t}\!\,#1}
\newcommand{\CC}{\mathbb{C}}
\newcommand{\PP}{\mathbb{P}}

\newcommand{\ZZ}{\mathbb{Z}}
\newcommand{\calP}{\mathcal{P}}
\newcommand{\ctv}{\mathrm{ctv}}
\newcommand{\cl}{\mathrm{cl}}
\newcommand{\bsx}{\boldsymbol{x}}
\newcommand{\bsy}{\boldsymbol{y}}

\begin{document}

\title{\bf Open string amplitudes of closed topological vertex}
\author{Kanehisa Takasaki\thanks{E-mail: takasaki@math.h.kyoto-u.ac.jp}\\
{\normalsize Department of Mathematics, Kinki University}\\ 
{\normalsize 3-4-1 Kowakae, Higashi-Osaka, Osaka 577-8502, Japan}\\
\\
Toshio Nakatsu\thanks{E-mail: nakatsu@mpg.setsunan.ac.jp}\\
{\normalsize Institute for Fundamental Sciences, Setsunan University}\\
{\normalsize 17-8 Ikeda Nakamachi, Neyagawa, Osaka 572-8508, Japan}}
\date{}
\maketitle 

\begin{abstract}
The closed topological vertex is the simplest ``off-strip'' case 
of non-compact toric Calabi-Yau threefolds with acyclic web diagrams. 
By the diagrammatic method of topological vertex, 
open string amplitudes of topological string theory therein 
can be obtained by gluing a single topological vertex 
to an ``on-strip'' subdiagram of the tree-like web diagram.  
If non-trivial partitions are assigned to just two parallel 
external lines of the web diagram, the amplitudes can be calculated 
with the aid of techniques borrowed from the melting crystal models.  
These amplitudes are thereby expressed as matrix elements, 
modified by simple prefactors, of an operator product on the Fock space 
of 2D charged free fermions.  This fermionic expression can be used 
to derive $q$-difference equations for generating functions of 
special subsets of the amplitudes.  These $q$-difference equations 
may be interpreted as the defining equation of a quantum mirror curve. 

\end{abstract}

\begin{flushleft}
\end{flushleft}

\newpage 

\section{Introduction}

Topological vertex \cite{AKMV03} is a diagrammatic method 
that captures A-model topological string theory  
on non-compact toric Calabi-Yau threefolds.  
In the case of ``on-strip'' geometry 
(see Appendix A for a precise setup), 
this method works particularly well to calculate 
both the closed string partition function and 
open string amplitudes in an explicit form \cite{IKP04}. 
Since the calculation in the on-strip case relies 
on the {\it linear} shape of the toric diagram, 
it is a technical challenge to extend this result 
to an ``off-strip'' case.   

The closed topological vertex \cite{BK03} 
is one of the simplest examples of ``off-strip'' geometry. 
Its web diagram is acyclic (in other words, 
the threefold has no compact 4-cycle), 
and the toric diagram has a triangular shape 
(see Figure \ref{fig1}). 
The closed string partition function in this case 
is calculated by several methods including 
topological vertex \cite{BK03,KLM05,Sulkowski06}. 
The final expression of the partition function 
resembles the on-strip case, but the method of derivation 
is more subtle.  Moreover, Karp, Liu and Mari\~no 
\cite[Section 6.4]{KLM05} argued that 
such a closed expression of the partition function 
will cease to exist if branches of the tree-like web diagram 
are prolonged to arbitrary lengths. In this sense, 
the closed topological vertex is rather special 
among off-strip geometry without compact 4-cycle. 

In this paper, we calculate open string amplitudes 
of the closed topological vertex in the case where 
non-trivial boundary conditions of the string world sheet 
are imposed on two parallel external lines of the web diagram. 
Although lacking full generality, this is the first attempt 
in the literature to calculate open string amplitudes 
of the closed topological vertex explicitly.  
Moreover, we use this result to derive $q$-difference equations 
for generating functions of special subsets 
of these amplitudes.  In the perspectives of mirror geometry 
of topological string theory \cite{ADKMV03,DV07}, 
the $q$-difference equations may be interpreted 
as the defining equations of a ``quantum mirror curve''.  
This quantum mirror curve will be a new example of quantum curves 
in the topological recursion program \cite{GS11}. 

\begin{figure}[h]
\centering
\includegraphics[scale=0.7]{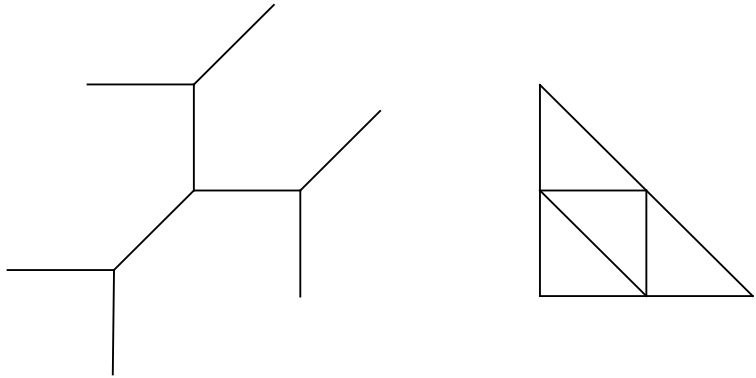}
\caption{Web diagram (left) and toric diagram (right) 
of closed topological vertex}
\label{fig1}
\end{figure}

To calculate the open string amplitudes in question, 
we use techniques that were developed 
in our previous work on the melting crystal models 
\cite{NT07,NT08,Takasaki13,Takasaki14}. 
A clue of these techniques is the notion of 
``shift symmetries'' in a quantum torus algebra.  
This algebra is realized by operators on the Fock space 
of 2D charged free fermions\footnote{
The same fermionic realization of 
the quantum torus algebra appears in the work 
of Okounkov and Pandharipande \cite{OP03} 
on the Gromov-Witten invariants of $\CC\PP^1$.}.  
The shift symmetries act on a set of basis elements $V^{(k)}_m$ 
of this algebra so as to shift the indices $k,m$ 
in a certain way.  This enables us to relate 
the commutative subalgebra spanned by $V^{(k)}_0$'s\footnote{
Its role as symmetries in the KP and 2D Toda hierarchies 
was independently studied by Harnad and Orlov \cite{HO09}.}
to the $U(1)$ current algebra spanned by $V^{(0)}_m$'s.  
In our previous work, this algebraic machinery is used 
to convert the partition functions of the melting crystal models 
to tau functions of the 2D Toda hierarchy. 
In this paper, we employ the same method to express 
the open string amplitudes as matrix elements, 
modified by simple prefactors, of an operator product 
on the fermionic Fock space.  

Our calculation starts from a {\it cut-and-glue} description 
of the amplitude \cite{Sulkowski06}.  Namely, 
the web diagram is cut into two subdiagrams by removing 
an internal line, and glued together along this line 
after calculating the contributions of these two parts. 
One of them is a single topological vertex, 
and the other is an on-strip diagram for which 
the result of Iqbal and Kashani-Poor \cite{IKP04} 
can be used. To glue these two parts again, 
we have to calculate an infinite sum with respect 
to a partition on the internal line. This is the place 
where the aforementioned techniques are used.  
The amplitude of the closed topological vertex thereby 
boils down to a product of simple factors and 
a matrix element of an operator on the Fock space.  
Moreover, the matrix element turns out to be 
the open string amplitude of a new on-strip web diagram. 

The final expression of the open string amplitudes 
enables us to derive $q$-difference equations 
for the generating functions of special subsets 
of the amplitudes.  The generating functions 
are the Baker-Akhiezer functions 
in the context of integrable hierarchies, 
and play the role of ``wave functions''  
of a probe D-brane \cite{ADKMV03,DV07}. 
Our result is an extension of known results 
on the resolved conifold \cite{KP06,HY06,Zhou12} 
and more general on-strip geometry \cite{Takasaki13bis}. 
The structure of the $q$-difference equation is, 
so to speak, a mixture of the $q$-difference equations 
of the quantum dilogarithmic functions \cite{FV93,FK93} 
and the $q$-hypergeometric equations that appear 
in the resolved conifold and more general on-strip geometry.  
Our result shows that quantum mirror curves beyond 
on-strip geometry can have an intricate origin.  

This paper is organized as follows.  In Section 2, 
the diagrammatic construction of the open string amplitudes 
are reformulated in a partially summed form.  Fermionic tools 
for the subsequent calculation are also reviewed here.  
In Section 3, the techniques borrowed from the melting crystal models 
are used to calculate the amplitudes in terms of fermions. 
In Section 4, the fermionic expression of the amplitudes 
is further converted to a final form.  A technical clue therein 
is the cyclic symmetry among ``two-leg'' topological vertices. 
This well known symmetry is translated to a kind of 
``operator-state correspondence'' in the fermionic Fock space, 
and used to rewrite the fermionic expression of the amplitudes. 
In Section 5, the generating functions of special subsets 
of the amplitudes are introduced, and shown 
to satisfy $q$-difference equations.  The structure 
of the $q$-difference equations is examined 
in the perspectives of mirror geometry.  
In Section 6, these results are shown to be consistent 
with a flop transition.  Appendix A is a brief review 
of the notion of on-strip amplitudes.  Appendix B presents 
another proof of the identities used in Section 4.

\section{Construction of open string amplitudes}

The setup for the open string amplitudes in question is shown 
in Figure \ref{fig2}.  $Q_1,Q_2,Q_3$ are K\"ahler parameters 
on the internal lines.  $\beta_1$ and $\beta_2$ are partitions 
assigned to the two lower external lines.  
The other external lines are given the trivial partition $\emptyset$.  
Let $Z^{\ctv}_{\beta_1\beta_2}$ denote the amplitude in this setup.  
$Z^{\ctv}_{\beta_1\beta_2}$ is a sum of weights over all possible values 
of the partitions $\alpha_1,\alpha_2,\alpha_3$ on the internal lines.  
The weight for a given configuration of $\alpha_1,\alpha_2,\alpha_3$ 
is a product of vertex weights and edge weights.  
These weights depend on the parameter $q$ in the range $|q|<1$.  

\begin{figure}
\centering
\includegraphics[scale=0.8]{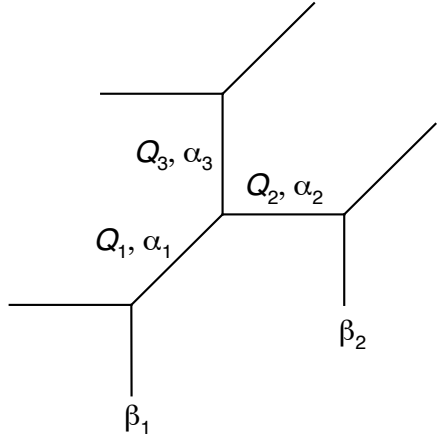}
\caption{Setup for open string amplitude $Z^{\ctv}_{\beta_1\beta_2}$}
\label{fig2}
\end{figure}

\subsection{Vertex weights and gluing rules}

The vertex weight at each vertex is 
the topological vertex\footnote{
We follow a definition commonly used 
in the recent literature \cite{IKV07,Taki07}.  
This definition differs from the earlier one 
\cite{AKMV03,ADKMV03} in that $q$ is replaced by $q^{-1}$ 
and an overall factor of the form 
$q^{\kappa(\lambda)/2+\kappa(\mu)/2+\kappa(\nu)/2}$ 
is multiplied.}
\beq
  C_{\lambda\mu\nu} = q^{\kappa(\mu)/2}s_{\tp{\nu}}(q^{-\rho})
    \sum_{\eta\in\calP}s_{\tp{\lambda}/\eta}(q^{-\nu-\rho})
      s_{\mu/\eta}(q^{-\tp{\nu}-\rho}), 
  \label{Clmn}
\eeq
where the sum with respect to $\eta$ ranges over the set $\calP$ 
of all partitions.  $\lambda = (\lambda_i)_{i=1}^\infty$, 
$\mu = (\mu_i)_{i=1}^\infty$ and $\nu = (\nu_i)_{i=1}^\infty$ 
are the partitions assigned to the three legs 
of the vertex that are ordered anti-clockwise, 
and $\tp{\nu}$ denotes the conjugate (or transposed) partition of $\nu$.  
$\kappa(\mu)$ is the second Casimir invariant 
\beqnn
  \kappa(\mu) = \sum_{i=1}^\infty \mu_i(\mu_i-2i+1) 
  = \sum_{i=1}^\infty\left((\mu_i-i+1/2)^2 - (-i+1/2)^2\right). 
\eeqnn
$s_{\tp{\nu}}(q^{-\rho})$, $s_{\tp{\lambda}/\eta}(q^{-\nu-\rho})$ and 
$s_{\mu/\eta}(q^{-\tp{\nu}-\rho})$ are special values 
of the infinite-variate Schur function $s_{\tp{\nu}}(\bsx)$ 
and the skew Schur functions $s_{\tp{\lambda}/\eta}(\bsx),\,s_{\mu/\eta}(\bsx)$, 
$\bsx = (x_1,x_2,\ldots)$, at 
\beqnn
  q^{-\rho} = (q^{i-1/2})_{i=1}^\infty,\quad 
  q^{-\nu-\rho} = (q^{-\nu_i+i-1/2})_{i=1}^\infty, \quad 
  q^{-\tp{\nu}-\rho} = (q^{-\tp{\nu}_i+i-1/2})_{i=1}^\infty. 
\eeqnn
The vertex weight enjoy the cyclic symmetry 
\beq
  C_{\lambda\mu\nu} = C_{\mu\nu\lambda} = C_{\nu\lambda\mu} 
  \label{CS}
\eeq
that can be deduced from the crystal interpretation 
of the vertex weight \cite{ORV03}.  

The vertex weights $C_{\lambda\mu\nu}$ and $C_{\lambda'\mu'\nu'}$ 
at two vertices connecting an internal line are glued together 
by the following rules: 
\begin{itemize}
\item[(i)] The partitions on the internal line,  
say $\lambda$ and $\lambda'$, are matched as 
\beqnn
  \lambda' = \tp{\lambda}. 
\eeqnn
\item[(ii)] The product of the vertex weights 
is multiplied by the edge weight 
\beqnn
  (-Q)^{|\lambda|}(-1)^{n|\lambda|}q^{-n\kappa(\lambda)/2}, 
\eeqnn
where $Q$ is the K\"ahler parameter of the internal line, 
and $n$ is an integer called ``the framing number''. 
\end{itemize}

The framing number is defined as 
\beq
  n = v'\wedge v = w'\wedge w, 
  \label{f-number}
\eeq
where $v,w$ and $v',w'$ are vectors in the web diagram 
that emanate from the two vertices (see Figure \ref{fig3}). 
The wedge product means the determinant 
of the $2\times 2$ matrix formed by the two vectors, 
i.e., 
\beqnn
  v'\wedge  v = v'_1v_2 - v'_2v_1 
\eeqnn
for $v' = (v'_1,v'_2)$ and $v = (v_1,v_2)$. 
These vectors $v,w$ and $v',w'$ are chosen along 
with the third vectors $u,u'$, $u + u' = 0$, 
in such a way that $u,v,w$ and $u',v',w'$ are ordered 
anti-clockwise and satisfy the zero-sum relations 
\beqnn
  u + v + w = 0,\quad u' + v' + w' = 0. 
\eeqnn
These sets of vectors are uniquely determined 
as far as the toric diagram is fully triangulated 
(i.e., the area of each triangle is $1/2$).  

\begin{figure}
\centering
\includegraphics[scale=0.7]{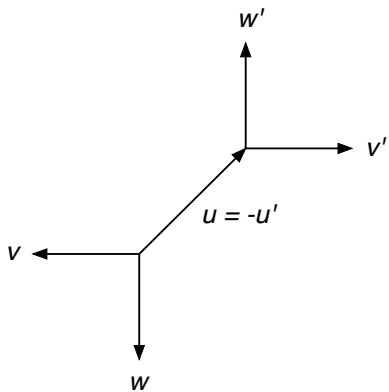}
\caption{Configuration of vectors that determine the framing number}
\label{fig3}
\end{figure}

\subsection{Reformulation of amplitude}

The amplitude $Z^{\ctv}_{\beta_1\beta_2}$ is given by a sum 
of the product of these weights 
over $\alpha_1,\alpha_2,\alpha_3 \in \calP$.  
Following Su{\l}kowski's formulation \cite{Sulkowski06}, 
we decompose this sum to a partial with 
sum respect to $\alpha_1,\alpha_2$ at the first stage 
and a sum with respect to $\alpha_3$ at the next stage.  
The full amplitude can be thus reformulated as 
\beq
  Z^{\ctv}_{\beta_1\beta_2} 
  = \sum_{\alpha_3\in\calP} Z_{\beta_1\beta_2|\alpha_3} 
    (-Q_3)^{|\alpha_3|}C_{\tp{\alpha}_3\emptyset\emptyset}. 
  \label{Zb1b2-fac}
\eeq
$Z_{\beta_1\beta_2|\alpha_3}$ is the partial sum with respect 
to $\alpha_1,\alpha_2$ and represents the contribution 
from the lower part of the web diagram.  This part is glued 
with the upper part via the internal line carrying $\alpha_3$.  
$(-Q_3)^{|\alpha_3|}$ is the edge weight of this internal line. 
Note that the framing number (\ref{f-number}) 
in this case is equal to $0$.  
$C_{\tp{\alpha}_3\emptyset\emptyset}$ is the contribution 
from the upper part of the web diagram.  
By the cyclic symmetry (\ref{CS}), 
this vertex weight reduces to a special value 
of the Schur function: 
\beq
  C_{\tp{\alpha}_3\emptyset\emptyset} 
  = C_{\emptyset\emptyset\tp{\alpha}_3} 
  = s_{\alpha_3}(q^{-\rho}). 
  \label{Ca3}
\eeq

The partial sum $Z_{\beta_1\beta_2|\alpha_3}$ itself 
may be thought of as an open string amplitude 
of the web diagram (called ``double-$\PP^1$'') 
shown in Figure \ref{fig4}. 
Since this is a diagram ``on a strip'', 
the associated open string amplitude can be calculated 
by the well known result \cite{IKP04} (see Appendix A): 
\begin{align}
  Z_{\beta_1\beta_2|\alpha_3} 
  &= s_{\tp{\beta}_1}(q^{-\rho})s_{\tp{\beta}_2}(q^{-\rho})s_{\tp{\alpha}_3}(q^{-\rho})
     \prod_{i,j=1}^\infty(1 - Q_1Q_2q^{-\beta_{1i}-\tp{\beta}_{2j}+i+j-1})^{-1}
     \notag\\
  &\quad\mbox{}\times 
     \prod_{i,j=1}^\infty(1 - Q_1q^{-\beta_{1i}-\alpha_{3j}+i+j-1})
     \prod_{i,j=1}^\infty(1 - Q_2q^{-\tp{\alpha}_{3i}-\tp{\beta}_{2j}+i+j-1}). 
  \label{Zb1b2a3}
\end{align}

\begin{figure}
\centering
\includegraphics[scale=0.8]{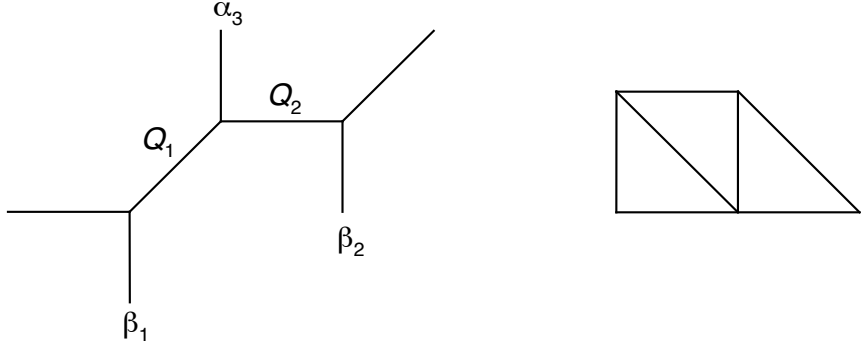}
\caption{Double-$\PP^1$ diagram defining $Z_{\beta_1\beta_2|\alpha_3}$}
\label{fig4}
\end{figure}

Plugging these building blocks into (\ref{Zb1b2-fac}), 
we obtain the following expression of $Z^{\ctv}_{\beta_1\beta_2}$: 
\begin{align}
  Z^{\ctv}_{\beta_1\beta_2} 
  &= s_{\tp{\beta}_1}(q^{-\rho})s_{\tp{\beta_2}}(q^{-\rho})
     \prod_{i,j=1}^\infty(1 - Q_1Q_2q^{-\beta_{1i}-\tp{\beta}_{2j}+i+j-1})^{-1}
     \notag\\
  &\quad\mbox{}\times \sum_{\alpha_3\in\calP}
     s_{\tp{\alpha}_3}(q^{-\rho})s_{\alpha_3}(q^{-\rho})(-Q_3)^{|\alpha_3|} 
     \prod_{i,j=1}^\infty(1 - Q_1q^{-\beta_{1i}-\alpha_{3j}+i+j-1})
     \notag\\
  &\quad\quad\quad\mbox{}\times 
        \prod_{i,j=1}^\infty(1 - Q_2q^{-\tp{\alpha}_{3i}-\tp{\beta}_{2j}+i+j-1}). 
  \label{Zb1b2-red}
\end{align}
Note here that the sum with respect to $\alpha_3$ resembles 
the partition function of the modified melting model 
\cite{Takasaki13,Takasaki14}:  The main part of 
the Boltzmann weight therein takes 
the product form $s_{\tp{\alpha}_3}(q^{-\rho})s_{\alpha_3}(q^{-\rho})$, 
and this weight is deformed by external potentials depending 
on $\alpha_3$.  To calculate this sum, we use the machinery 
of 2D charged free fermions.

\subsection{Fermionic Fock space and operators} 

The setup of the fermionic Fock space and operators 
is the same as used for the melting crystal models 
\cite{NT07,NT08,Takasaki13,Takasaki14}. 
Let $\psi_n,\psi^*_n$, $n\in\ZZ$, denote the Fourier modes 
of the 2D charged free fermion fields $\psi(z),\psi^*(z)$. 
They satisfy the anti-commutation relations 
\beqnn
  \psi_m\psi^*_n + \psi^*_n\psi_m = \delta_{m+n,0}, \quad
  \psi_m\psi_n + \psi_n\psi_m = 0, \quad 
   \psi^*_m\psi^*_n + \psi^*_n\psi^*_m = 0. 
\eeqnn
The associated Fock space and its dual space 
are decomposed to the charge-$s$ sectors for $s \in \ZZ$.  
It is only the charge-$0$ sector that is relevant to 
the calculation of (\ref{Zb1b2-red}).  An orthonormal basis 
of the charge-$0$ sector is given the ground states 
\beqnn
\begin{aligned}
  \langle 0| &= \langle -\infty|\cdots\psi^*_{-i+1}\cdots\psi^*_{-1}\psi^*_0,\\
  |0\rangle &= \psi_0\psi_1\cdots\psi_{i-1}\cdots|-\infty\rangle 
\end{aligned}
\eeqnn
and the excited states 
\beqnn
\begin{aligned}
  \langle \lambda| &= \langle -\infty|\cdots\psi^*_{\lambda_i-i+1}
    \cdots\psi^*_{\lambda_2-1}\psi^*_{\lambda_1},\\
  |\lambda\rangle &= \psi_{-\lambda_1}\psi_{-\lambda_2+1}\cdots
    \psi_{-\lambda_i+i-1}\cdots|-\infty\rangle 
\end{aligned}
\eeqnn
labelled by partitions.  The normal ordered product 
${:}\psi_m\psi^*_n{:}$ is defined as 
\beqnn
  {:}\psi_m\psi^*_n{:} 
  = \psi_m\psi^*_n - \langle 0|\psi_m\psi^*_n|0\rangle. 
\eeqnn

The following operators on the Fock space are used 
as fundamental tools in our calculation. 
\begin{itemize}
\item[(i)] The zero-modes 
\beqnn
  L_0 = \sum_{n\in\ZZ}n{:}\psi_{-n}\psi^*_n{:},\quad
  W_0 = \sum_{n\in\ZZ}n^2{:}\psi_{-n}\psi^*_n{:} 
\eeqnn
of the Virasoro and $W_3$ algebras and the Fourier modes 
\beqnn
  J_m = \sum_{n\in\ZZ}{:}\psi_{-n}\psi^*_{n+m}{:},\quad m\in\ZZ, 
\eeqnn
of the fermionic current ${:}\psi(z)\psi^*(z){:}$. 
\item[(ii)] The fermionic realization 
\beqnn
  K = \sum_{n\in\ZZ}(n-1/2)^2{:}\psi_{-n}\psi^*_n{:} 
    = W_0 - L_0 + J_0/4 
\eeqnn
of the so called ``cut-and-join operator'' \cite{GJ97,KL06}. 
\item[(iii)] The basis elements 
\beqnn
  V^{(k)}_m 
  = q^{-km/2}\sum_{n\in\ZZ}q^{kn}{:}\psi_{m-n}\psi^*_{n}{:},\quad 
  k,m\in\ZZ, 
\eeqnn
of a fermionic realization of the quantum torus algebra 
\cite{NT07,OP03}. 
\item[(iv)] The vertex operators \cite{OR01,BY08} 
\beqnn
  \Gamma_{\pm}(z) 
  = \exp\left(\sum_{k=1}^\infty\frac{z^k}{k}J_{\pm k}\right),\quad
  \Gamma'_{\pm}(z) 
  = \exp\left(- \sum_{k=1}^\infty\frac{(-z)^k}{k}J_{\pm k}\right) 
\eeqnn
and the multi-variable extensions 
\beqnn
  \Gamma_{\pm}(\bsx) = \prod_{i\ge 1}\Gamma_{\pm}(x_i),\quad 
  \Gamma'_{\pm}(\bsx) = \prod_{i\ge 1}\Gamma'_{\pm}(x_i). 
\eeqnn
\end{itemize}

The matrix elements of these operators are well known. 
$J_0,L_0,W_0,K$ are diagonal with respect to 
the basis $\{|\lambda\rangle\}_{\lambda\in\calP}$ 
in the charge-$0$ sector: 
\begin{gather}
  \langle\lambda|J_0|\mu\rangle = 0,\quad
  \langle\lambda|L_0|\mu\rangle = \delta_{\lambda\mu}|\lambda|,\notag\\
  \langle\lambda|W_0|\mu\rangle 
    = \delta_{\lambda\mu}\left(\kappa(\lambda) + |\lambda|\right),\quad
  \langle\lambda|K|\mu\rangle 
    = \delta_{\lambda\mu}\kappa(\lambda). 
\end{gather}
The matrix elements of $\Gamma_{\pm}(\bsx)$ and $\Gamma'_{\pm}(\bsx)$ 
are skew Schur functions \cite{Mac-book,MJD-book}: 
\begin{gather}
  \langle\lambda|\Gamma_{-}(\bsx)|\mu\rangle
  = \langle\mu|\Gamma_{+}(\bsx)|\lambda\rangle
  = s_{\lambda/\mu}(\bsx), \notag\\
  \langle\lambda|\Gamma'_{-}(\bsx)|\mu\rangle 
  = \langle\mu|\Gamma'_{+}(\bsx)|\lambda\rangle
  = s_{\tp{\lambda}/\tp{\mu}}(\bsx). 
\end{gather}

\section{Calculation of sum in (\ref{Zb1b2-red})}

Let us proceed to calculation of the sum in (\ref{Zb1b2-red}).  
This comprises two steps.  In the first step, 
we express $s_{\tp{\alpha}_3}(q^{-\alpha})$ and 
$s_{\alpha_3}(q^{-\rho})$ in a fermionic form, 
and convert the c-number factors $\prod_{i,j=1}^\infty(1 - Q_1\cdots)$ 
and $\prod_{i,j=1}^\infty(1 - Q_2\cdots)$ to operators 
inserted in the fermionic expression of the Schur functions.  
The sum with respect to $\alpha_3$ thereby 
turns into the vacuum expectation value 
of an operator product on the Fock space.  
In the second step, we use the ``shift symmetries'' 
of the quantum torus algebra 
\cite{NT07,NT08,Takasaki13,Takasaki14} 
to rewrite the vacuum expectation value further.  
This calculation is more or less parallel to the way 
the partition functions of the various melting crystal models 
are converted to tau functions of the 2D Toda hierarchy.

\subsection{Step 1: Translation to fermionic language}

The infinite products $\prod_{i,j}^\infty(1 - Q_1\cdots)$ 
and $\prod_{i,j=1}^\infty (1 - Q_2\cdots)$ 
can be re-expressed in an exponential form as 
\beqnn
  \prod_{i,j=1}^\infty(1- Q_1q^{-\beta_{1i}-\alpha_{3j}+i+j-1}) 
  = \exp\left(- \sum_{i,k=1}^\infty
      \frac{(Q_1q^{-\beta_{1i}+i})^k}{k} 
      \sum_{j=1}^\infty q^{-k(\alpha_{3j}-j+1)} \right) 
\eeqnn
and 
\beqnn
  \prod_{i,j=1}^\infty(1 - Q_2q^{-\tp{\alpha}_{3i}-\tp{\beta}_{2j}+i+j-1}) 
  = \exp\left(- \sum_{j,k=1}^\infty
       \frac{(Q_2q^{-\tp{\beta}_{2j}+j})^k}{k} 
       \sum_{i=1}^\infty q^{-k(\tp{\alpha}_{3i}-i+1)} \right). 
\eeqnn
We convert these c-number factors to operators 
inserted in the fermionic expression 
\beqnn
  s_{\tp{\alpha}_3}(q^{-\rho}) 
  = \langle 0|\Gamma'_{+}(q^{-\rho})|\alpha_3\rangle,\quad 
  s_{\alpha_3}(q^{-\rho}) 
  = \langle\alpha_3|\Gamma_{-}(q^{-\rho})|0\rangle 
\eeqnn
of the special values of the Schur functions.  

To this end, let us note that 
$\sum_{j=1}^\infty q^{-k(\alpha_{3j}-j+1)}$ and 
$\sum_{j=1}^\infty q^{-k(\tp{\alpha}_{3i}-i+1)}$ are 
related to eigenvalues of $V^{(\pm k)}_0$'s as shown below.  

\begin{lemma} For any $k > 0$ and any $\lambda\in\calP$, 
\begin{align}
  \left(V^{(-k)}_0 + \frac{1}{1-q^k}\right)|\lambda\rangle 
  &= \sum_{i=1}^\infty q^{-k(\lambda_i-i+1)}|\lambda\rangle,
    \label{V(-k)0-action}\\
  \left(V^{(k)}_0 - \frac{q^k}{1-q^k}\right)|\lambda\rangle 
  &=  - q^k \sum_{i=1}^\infty q^{-k(\tp{\lambda}_i-i+1)}|\lambda\rangle. 
    \label{V(k)0-action}
\end{align}
\end{lemma}

\begin{remark}
(\ref{V(-k)0-action}) and (\ref{V(k)0-action}) 
imply the relations 
\beqnn
\begin{aligned}
  \langle\lambda|\left(V^{(-k)}_0 + \frac{1}{1-q^k}\right)
  &= \langle\lambda|\sum_{i=1}^\infty q^{-k(\lambda_i-i+1)},\\
  \langle\lambda|\left(V^{(k)}_0 - \frac{q^k}{1-q^k}\right)
  &=  - \langle\lambda|q^k \sum_{i=1}^\infty q^{-k(\tp{\lambda}_i-i+1)}
\end{aligned}
\eeqnn
in the dual Fock space as well. 
\end{remark}

\begin{proof}
It is straightforward to derive (\ref{V(-k)0-action}): 
\beqnn
\begin{aligned}
  V^{(-k)}_0|\lambda\rangle 
  &= \sum_{j=1}^\infty(q^{-k(\lambda_j-j+1)} - q^{-k(-j+1)})|\lambda\rangle\\
  &= \left(\sum_{j=1}^\infty q^{-k(\lambda_j-j+1)} - \frac{1}{1-q^k}
    \right)|\lambda\rangle. 
\end{aligned}
\eeqnn
The subtraction term $q^{-k(-j+1)}$ in this calculation 
originates in the normal ordering 
\beqnn
  {:}\psi_{-n}\psi^*_n{:} 
  = \begin{cases}
    \psi_{-n}\psi^*_n &\text{for $n > 0$},\\
    \psi_{-n}\psi^*_n - 1 &\text{for $n\leq 0$}. 
    \end{cases}
\eeqnn
It is not straightforward to derive (\ref{V(k)0-action}). 
Let $n$ be an integer greater than or equal 
to the length of $\lambda$.  
Accordingly, $\lambda_i = i$ for $i > n$.  
Since the set of all integers $i\leq n$ 
can be divided into two disjoint sets as 
\beqnn
  \{i\mid i\leq n\} 
  = \{\tp{\lambda}_i-i+1 \mid i\geq 1\} 
    \cup \{-\lambda_i+i \mid 1\leq i\leq n\}, 
\eeqnn
one obtains the identity 
\beqnn
  \sum_{i=1}^\infty q^{-k(\tp{\lambda}_i-i+1)} 
  + \sum_{i=1}^n q^{-k(-\lambda_i+i)} 
  = \sum_{i=-\infty}^n q^{-ki} 
  = \frac{q^{-kn}}{1-q^k}, 
\eeqnn
which implies that 
\beqnn
\begin{aligned}
  \sum_{i=1}^\infty q^{-k(\tp{\lambda}_i-i+1)} 
  &= - \sum_{i=1}^n q^{-k(-\lambda_i+i)} + \frac{q^{-kn}}{1-q^k}\\
  &= - \sum_{i=1}^n (q^{-k(-\lambda_i+i)} - q^{-ki}) 
     + \frac{1}{1-q^k} \\
  &= - q^{-k}\sum_{i=1}^n(q^{k(\lambda_i-i+1)} - q^{k(-i+1)}) 
     + \frac{1}{1-q^k}. 
\end{aligned}
\eeqnn
Consequently, 
\beqnn
\begin{aligned}
  V^{(k)}_0|\lambda\rangle 
  &= \sum_{i=1}^n (q^{k(\lambda_i-i+1)} - q^{k(-i+1)})|\lambda\rangle\\
  &= \left(- q^k\sum_{i=1}^\infty q^{-k(\tp{\lambda}_i-i+1)} 
       + \frac{q^k}{1-q^k}\right) |\lambda\rangle. 
\end{aligned}
\eeqnn
(\ref{V(k)0-action}) can be thus derived. 
\end{proof}

By (\ref{V(-k)0-action}) and (\ref{V(k)0-action}), 
the c-number factors $\prod_{i,j=1}^\infty(1 - Q_1\cdots)$ 
and $\prod_{i,j=1}^\infty(1 - Q_2\cdots)$ can be converted 
to operators on the Fock space as 
\begin{align}
  &\prod_{i,j=1}^\infty(1 - Q_1q^{-\beta_{1i}-\alpha_{3j}+i+j-1})
   \cdot s_{\tp{\alpha}_3}(q^{-\rho}) \notag\\
  &= \langle 0|\Gamma'_{+}(q^{-\rho}) 
       \exp\left(- \sum_{i,k=1}^\infty\frac{(Q_1q^{-\beta_{1i}+i})^k}{k} 
       \left(V^{(-k)}_0 + \frac{1}{1-q^k}\right) \right) 
     |\alpha_3\rangle
  \label{a3-block1}
\end{align}
and
\begin{align}
  &\prod_{i,j=1}^\infty(1 - Q_2q^{-\tp{\alpha}_{3i}-\tp{\beta}_{2j}+i+j-1})
   \cdot s_{\alpha_3}(q^{-\rho}) \notag\\
  &= \langle\alpha_3|
       \exp\left(\sum_{j,k=1}^\infty\frac{(Q_2q^{-\tp{\beta}_{2j}+j-1})^k}{k} 
       \left(V^{(k)}_0 - \frac{q^k}{1-q^k}\right) \right)
     \Gamma_{-}(q^{-\rho})|0\rangle. 
  \label{a3-block2}
\end{align}
Moreover, the factor $(-Q_3)^{|\alpha_3|}$ can be identified 
with the diagonal matrix element of $(-Q_3)^{L_0}$. 
Having derived these building blocks, we can now use 
the partition of unity 
\beqnn
  \sum_{\alpha_3\in\calP}|\alpha_3\rangle\langle\alpha_3| 
  = 1 
\eeqnn
in the charge-$0$ sector to rewrite the sum 
in (\ref{Zb1b2-red}) to the vacuum expectation value 
of an operator product: 
\begin{align}
  &\sum_{\alpha_3\in\calP}s_{\tp{\alpha_3}}(q^{-\rho})s_{\alpha_3}(q^{-\rho})
   (-Q_3)^{|\alpha_3|}\prod_{i,j=1}^\infty(1 - Q_1\cdots)
   \prod_{i,j=1}^\infty(1 - Q_2\cdots) \notag\\
  &= \langle 0|\Gamma'_{+}(q^{-\rho}) 
       \exp\left(- \sum_{i,k=1}^\infty\frac{(Q_1q^{-\beta_{1i}+i})^k}{k} 
       \left(V^{(-k)}_0 + \frac{1}{1-q^k}\right) \right)
       (-Q_3)^{L_0}  \notag\\
  &\quad\mbox{}\times 
       \exp\left(\sum_{j,k=1}^\infty\frac{(Q_2q^{-\tp{\beta}_{2j}+j-1})^k}{k} 
       \left(V^{(k)}_0 - \frac{q^k}{1-q^k}\right) \right)
     \Gamma_{-}(q^{-\rho})|0\rangle. 
  \label{a3sum=vev1}
\end{align}

\subsection{Step 2: Use of shift symmetries}

Let us recall the following consequence of 
the shift symmetries of the quantum torus algebra 
\cite{NT07,NT08,Takasaki13}.  Note that the last one 
(\ref{SS3}) is modified from the previous formulation 
in terms of $W_0$. 

\begin{lemma} 
\begin{gather}
  \Gamma'_{-}(q^{-\rho})\Gamma'_{+}(q^{-\rho})
  \left(V^{(-k)}_0 + \frac{1}{1-q^k}\right) 
    = V^{(-k)}_k\Gamma'_{-}(q^{-\rho})\Gamma'_{+}(q^{-\rho}),
  \label{SS1}\\
  \left(V^{(k)}_0 - \frac{q^k}{1-q^k}\right) 
  \Gamma_{-}(q^{-\rho})\Gamma_{+}(q^{-\rho})
    =   \Gamma_{-}(q^{-\rho})\Gamma_{+}(q^{-\rho})(-1)^kV^{(k)}_{-k},
  \label{SS2}\\
  V^{(-k)}_k = q^{-k/2}q^{K/2}J_kq^{-K/2},\quad 
  V^{(k)}_{-k} = q^{k/2}q^{K/2}J_{-k}q^{-K/2}. 
  \label{SS3}
\end{gather}
\end{lemma}

We use these operator identities to rewrite 
(\ref{a3sum=vev1}) further.  Let us first examine 
the left side of $(-Q_3)^{L_3}$ in (\ref{a3sum=vev1}).  
Upon inserting $q^{-K/2}\Gamma'_{-}(q^{-\rho})$ 
to the right of $\langle 0|$ as 
\beqnn
  \langle 0|\Gamma'_{+}(q^{-\rho}) 
  = \langle 0|q^{-K/2}\Gamma'_{-}(q^{-\rho})\Gamma'_{+}(q^{-\rho}), 
\eeqnn
we can use (\ref{SS1}) and (\ref{SS3}) to rewrite 
the left side of $(-Q_3)^{L_3}$ as 
\beqnn
\begin{aligned}
  &\langle 0|\Gamma'_{+}(q^{-\rho}) 
   \exp\left(- \sum_{i,k=1}^\infty\frac{(Q_1q^{-\beta_{1i}+i})^k}{k} 
   \left(V^{(-k)}_0 + \frac{1}{1-q^k}\right) \right)\\
  &= \langle 0|q^{-K/2}\exp\left(- \sum_{i,k=1}^\infty
     \frac{(Q_1q^{-\beta_{1i}+i})^k}{k} V^{(-k)}_k\right)
     \Gamma'_{-}(q^{-\rho})\Gamma'_{+}(q^{-\rho})\\
  &= \langle 0|\exp\left(- \sum_{i,k=1}^\infty
     \frac{(Q_1q^{-\beta_{1i}+i})^k}{k}q^{-k/2}J_k\right)
     q^{-K/2}\Gamma'_{-}(q^{-\rho})\Gamma'_{+}(q^{-\rho})\\
  &= \langle 0|\exp\left(- \sum_{i,k=1}^\infty
     \frac{(Q_1q^{-\beta_{1i}+i-1/2})^k}{k}J_k\right)
     q^{-K/2}\Gamma'_{-}(q^{-\rho})\Gamma'_{+}(q^{-\rho}). 
\end{aligned}
\eeqnn
The exponential operator in the last line is essentially 
a vertex operator, 
\beqnn
  \exp\left(- \sum_{i,k=1}^\infty
  \frac{(Q_1q^{-\beta_{1i}+i-1/2})^k}{k}J_k\right)
  = (-Q_1)^{-L_0}\Gamma'_{+}(q^{-\beta_1-\rho})(-Q_1)^{L_0}, 
\eeqnn
hence 
\begin{align}
  &\langle 0|\Gamma'_{+}(q^{-\rho}) 
   \exp\left(- \sum_{i,k=1}^\infty\frac{(Q_1q^{-\beta_{1i}+i})^k}{k} 
   \left(V^{(-k)}_0 + \frac{1}{1-q^k}\right) \right) \notag\\
  &= \langle 0|\Gamma'_{+}(q^{-\beta_1-\rho})(-Q_1)^{L_0}
     q^{-K/2}\Gamma'_{-}(q^{-\rho})\Gamma'_{+}(q^{-\rho}). 
  \label{vev1-left}
\end{align}

In exactly the same manner, using (\ref{SS2}) and (\ref{SS3}), 
we can rewrite the right side of $(-Q_3)^{L_0}$ as 
\begin{align}
  &\exp\left(\sum_{j,k=1}^\infty\frac{(Q_2q^{-\tp{\beta}_{2j}+j-1})^k}{k} 
   \left(V^{(k)}_0 - \frac{q^k}{1-q^k}\right) \right)
   \Gamma_{-}(q^{-\rho})|0\rangle \notag\\
  &= \Gamma_{-}(q^{-\rho})\Gamma_{+}(q^{-\rho})q^{K/2}
     (-Q_2)^{L_0}\Gamma_{-}(q^{-\tp{\beta}_2-\rho})|0\rangle. 
  \label{vev1-right}
\end{align}

Plugging (\ref{vev1-left}) and (\ref{vev1-right}) 
into (\ref{a3sum=vev1}) yields the following expression 
of the sum in (\ref{Zb1b2-red}): 
\begin{align}
  &\sum_{\alpha_3\in\calP}s_{\tp{\alpha}_3}(q^{-\rho})s_{\alpha_3}(q^{-\rho})
   (-Q_3)^{|\alpha_3|}\prod_{i,j=1}^\infty(1 - Q_1\cdots)
   \prod_{i,j=1}^\infty(1 - Q_2\cdots) \notag\\
  &= \langle 0|\Gamma'_{+}(q^{-\beta_1-\rho})(-Q_1)^{L_0}
     q^{-K/2}\Gamma'_{-}(q^{-\rho})\Gamma'_{+}(q^{-\rho})(-Q_3)^{L_0} \notag\\
  &\quad\mbox{}\times
     \Gamma_{-}(q^{-\rho})\Gamma_{+}(q^{-\rho})q^{K/2}
     (-Q_2)^{L_0}\Gamma_{-}(q^{-\tp{\beta}_2-\rho})|0\rangle. 
  \label{a3sum=vev2}
\end{align}

\section{Final expression of open string amplitudes}

We have thus derived the following intermediate expression 
of $Z^{\ctv}_{\beta_1\beta_2}$: 
\begin{align}
  Z^{\ctv}_{\beta_1\beta_2} 
  &= s_{\tp{\beta}_1}(q^{-\rho})s_{\tp{\beta}_2}(q^{-\rho})
     \prod_{i,j=1}^\infty(1 - Q_1Q_2q^{-\beta_{1i}-\tp{\beta}_{2j}+i+j-1})^{-1}
     \notag\\
  &\quad\mbox{}\times 
     \langle 0|\Gamma'_{+}(q^{-\beta_1-\rho})(-Q_1)^{L_0}
     q^{-K/2}\Gamma'_{-}(q^{-\rho})\Gamma'_{+}(q^{-\rho})(-Q_3)^{L_0} \notag\\
  &\quad\mbox{}\times
     \Gamma_{-}(q^{-\rho})\Gamma_{+}(q^{-\rho})q^{K/2}
     (-Q_2)^{L_0}\Gamma_{-}(q^{-\tp{\beta}_2-\rho})|0\rangle. 
   \label{Zb1b2-med}
\end{align}
As a final step, we use the following relations 
in the fermionic Fock space that can be derived 
from a special case of the cyclic symmetry (\ref{CS}). 
This is a kind of {\it operator-state correspondence\/} 
that maps vertex operators 
of the form $\Gamma_{\pm}(q^{-\tp{\lambda}-\rho})$ 
and $\Gamma'_{\pm}(q^{-\lambda-\rho})$ 
to the state vectors $\langle\tp{\lambda}|$ 
and $|\tp{\lambda}\rangle$ in the Fock space. 

\begin{lemma}
For any $\lambda\in\calP$, 
\begin{align}
  s_{\tp{\lambda}}(q^{-\rho})\Gamma'_{-}(q^{-\lambda-\rho})|0\rangle 
    &= q^{K/2}\Gamma_{-}(q^{-\rho})\Gamma_{+}(q^{-\rho})|\tp{\lambda}\rangle,
  \label{2leg-CS1}\\
  s_{\tp{\lambda}}(q^{-\rho})\Gamma_{-}(q^{-\tp{\lambda}-\rho})|0\rangle 
    &= q^{\kappa(\lambda)/2}q^{-K/2}\Gamma'_{-}(q^{-\rho})\Gamma'_{+}(q^{-\rho})
       |\tp{\lambda}\rangle.
  \label{2leg-CS2}
\end{align}
\end{lemma}

\begin{remark}
There are a number of apparently different, but equivalent 
forms of these relations.  For example, one can use 
the well known identity \cite{Mac-book}
\beq
  s_{\tp{\lambda}}(q^{-\rho}) 
  = q^{\kappa(\lambda)/2}s_\lambda(q^{-\rho})
  \label{tp-rel}
\eeq
to rewrite (\ref{2leg-CS2}) as 
\beq
  s_{\lambda}(q^{-\rho})\Gamma_{-}(q^{-\tp{\lambda}-\rho})|0\rangle 
    = q^{-K/2}\Gamma'_{-}(q^{-\rho})\Gamma'_{+}(q^{-\rho})
      |\tp{\lambda}\rangle. 
  \label{2leg-CS3}
\eeq
(\ref{2leg-CS3}), in turn, is equivalent to (\ref{2leg-CS1}) 
(with $\lambda$ being replaced by $\tp{\lambda}$) 
as one can see from the identities 
\beqnn
\begin{aligned}
  \langle\tp{\mu}|\Gamma_{-}(q^{-\tp{\lambda}-\rho})|0\rangle 
  &= \langle\mu|\Gamma'_{-}(q^{-\tp{\lambda}-\rho})|0\rangle, \notag\\
  \langle\tp{\mu}|\Gamma'_{-}(q^{-\rho})\Gamma'_{+}(q^{-\rho})
  |\tp{\lambda}\rangle
  &= \langle\mu|\Gamma_{-}(q^{-\rho})\Gamma_{+}(q^{-\rho})|\lambda\rangle 
\end{aligned}
\eeqnn
and the fact that $\langle\lambda|$ and $|\lambda\rangle$ 
are eigenvector of $K$ with eigenvalue $\kappa(\lambda)$.  
(\ref{2leg-CS1}), (\ref{2leg-CS2}) and (\ref{2leg-CS3}) imply the relations 
\beqnn
\begin{aligned}
  s_{\tp{\lambda}}(q^{-\rho})\langle 0|\Gamma'_{+}(q^{-\lambda-\rho})
    &= \langle\tp{\lambda}|\Gamma_{-}(q^{-\rho})\Gamma_{+}(q^{-\rho})q^{K/2}\\
  s_{\tp{\lambda}}(q^{-\rho})\langle 0|\Gamma_{+}(q^{-\tp{\lambda}-\rho})
    &= q^{\kappa(\lambda)/2}\langle\tp{\lambda}|
       \Gamma'_{-}(q^{-\rho})\Gamma'_{+}(q^{-\rho})q^{-K/2},\\
  s_{\lambda}(q^{-\rho})\langle 0|\Gamma_{+}(q^{-\tp{\lambda}-\rho})
    &= \langle\tp{\lambda}|\Gamma'_{-}(q^{-\rho})\Gamma'_{+}(q^{-\rho})q^{-K/2}
\end{aligned}
\eeqnn
in the dual Fock space as well.  
\end{remark}

\begin{proof}
The topological vertex has the fermionic expression 
\begin{align}
  C_{\lambda\mu\nu} 
  &= q^{\kappa(\mu)/2}s_{\tp{\nu}}(q^{-\rho}) 
     \langle\tp{\lambda}|\Gamma_{-}(q^{-\nu-\rho})
     \Gamma_{+}(q^{-\tp{\nu}-\rho})|\mu\rangle \notag\\
  &= q^{\kappa(\mu)/2}s_{\tp{\nu}}(q^{-\rho})
     \langle\lambda|\Gamma'_{-}(q^{-\nu-\rho})
     \Gamma'_{+}(q^{-\tp{\nu}-\rho})|\tp{\mu}\rangle. 
  \label{Cabc-fermion}
\end{align}
The ``two-leg'' case $C_{\mu\emptyset\lambda} = C_{\lambda\mu\emptyset}$ 
of the cyclic symmetry (\ref{CS})\footnote{Zhou \cite{Zhou03} 
gave a direct proof of the two-leg cyclic symmetry 
without relying on the crystal interpretation 
of Okounkov, Reshetikhin and Vafa \cite{ORV03}.  
We present another direct proof in Append B that employs 
the same techniques as used in Section 3.} 
thereby turns into the relation 
\beq
  s_{\tp{\lambda}}(q^{-\rho})
  \langle 0|\Gamma'_{+}(q^{-\lambda-\rho})|\mu\rangle
  = \langle\tp{\lambda}|\Gamma_{-}(q^{-\rho})\Gamma_{+}(q^{-\rho})
    q^{K/2}|\mu\rangle. 
  \label{2leg-CS4}
\eeq
among matrix elements of operators on the Fock space. 
Since this identity holds for any $\mu$, 
one obtains (\ref{2leg-CS1}) in the dual form. 
Similarly, the symmetry relation $C_{\emptyset\mu\lambda} 
= C_{\mu\lambda\emptyset}$ yields the identity 
\beq
  s_{\tp{\lambda}}(q^{-\rho})
  \langle\mu|q^{K/2}\Gamma_{-}(q^{-\tp{\lambda}-\rho})|0\rangle 
  = \langle\mu|\Gamma'_{-}(q^{-\rho})\Gamma'_{+}(q^{-\rho})
    |\tp{\lambda}\rangle q^{\kappa(\lambda)/2}, 
  \label{2leg-CS5}
\eeq
and this implies (\ref{2leg-CS2}). 
\end{proof}

We can use the specialization 
\beqnn
\begin{aligned}
  s_{\tp{\beta}_1}(q^{-\rho})\langle 0|\Gamma'_{+}(q^{-\beta_1-\rho})
  &= \langle\tp{\beta}_1|\Gamma_{-}(q^{-\rho})\Gamma_{+}(q^{-\rho})q^{K/2},\\
  s_{\tp{\beta}_2}(q^{-\rho})\Gamma_{-}(q^{-\tp{\beta}_2-\rho})|0\rangle
  &= q^{\kappa(\beta_2)/2}q^{-K/2}
    \Gamma'_{-}(q^{-\rho})\Gamma'_{+}(q^{-\rho})|\tp{\beta}_2\rangle 
\end{aligned}
\eeqnn
of (\ref{2leg-CS1}) and (\ref{2leg-CS2}) 
to $\lambda=\beta_1$ and $\lambda=\beta_2$ 
to rewrite (\ref{Zb1b2-med}) as 
\beqnn
\begin{aligned}
  Z^{\ctv}_{\beta_1\beta_2} 
  &= q^{\kappa(\beta_2)/2}\prod_{i,j=1}^\infty
     (1 - Q_1Q_2q^{-\beta_{1i}-\tp{\beta}_{2j}+i+j-1})^{-1}\\
  &\quad\mbox{}\times \langle\tp{\beta}_1|
     \Gamma_{-}(q^{-\rho})\Gamma_{+}(q^{-\rho})q^{K/2}(-Q_1)^{L_0}q^{-K/2}
     \Gamma'_{-}(q^{-\rho})\Gamma'_{+}(q^{-\rho})(-Q_3)^{L_0}\\
  &\quad\mbox{}\times 
     \Gamma_{-}(q^{-\rho})\Gamma_{+}(q^{-\rho})q^{K/2}(-Q_2)^{L_0}q^{-K/2}
     \Gamma'_{-}(q^{-\rho})\Gamma'_{+}(q^{-\rho})|\tp{\beta}_2\rangle. 
\end{aligned}
\eeqnn
Since $q^{K/2}$'s and $q^{-K/2}$'s in this expression cancel out as 
\beq
  q^{K/2}(-Q_1)^{L_0}q^{-K/2} = (-Q_1)^{L_0},\quad 
  q^{K/2}(-Q_2)^{L_0}q^{-K/2} = (-Q_2)^{L_0},
  \label{cancellation}
\eeq
we arrive at the following final expression of $Z^{\ctv}_{\beta_1\beta_2}$. 

\begin{theorem}
The open string amplitude $Z^{\ctv}_{\beta_1\beta_2}$ can be expressed as 
\begin{align}
  Z^{\ctv}_{\beta_1\beta_2} 
  &= q^{\kappa(\beta_2)/2}\prod_{i,j=1}^\infty
     (1 - Q_1Q_2q^{-\beta_{1i}-\tp{\beta}_{2j}+i+j-1})^{-1} \notag\\
  &\quad\mbox{}\times \langle\tp{\beta}_1|
     \Gamma_{-}(q^{-\rho})\Gamma_{+}(q^{-\rho})(-Q_1)^{L_0}
     \Gamma'_{-}(q^{-\rho})\Gamma'_{+}(q^{-\rho})(-Q_3)^{L_0} \notag\\
  &\quad\mbox{}\times
     \Gamma_{-}(q^{-\rho})\Gamma_{+}(q^{-\rho})
     (-Q_2)^{L_0}\Gamma'_{-}(q^{-\rho})\Gamma'_{+}(q^{-\rho})
     |\tp{\beta}_2\rangle. 
  \label{Zb1b2-fin}
\end{align}
\end{theorem}
  
Let us note here that the main part 
$\langle\tp{\beta}_1|\cdots|\tp{\beta_2}\rangle$ 
of this expression coincides with the open string amplitude 
of the on-strip web diagram shown in Figure \ref{fig5} 
(see Appendix A for general formulae of amplitudes).  
Thus, speaking schematically, 
gluing the one-leg vertex (\ref{Ca3}) 
to the on-strip web diagram of Figure \ref{fig4} 
generates another on-strip web diagram and its correction 
$q^{\kappa(\beta_2)/2}\prod_{i,j=1}^\infty(1-Q_1Q_2\cdots)$. 
This structure of (\ref{Zb1b2-fin}) is a key 
to derive $q$-difference equations for generating functions.

\begin{figure}
\centering
\includegraphics[scale=0.8]{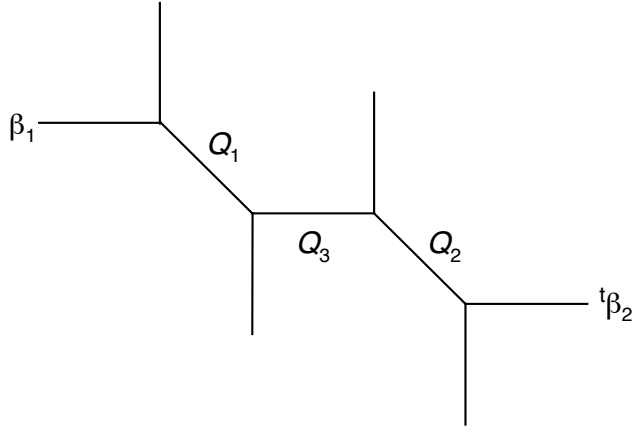}
\caption{Web diagram emerging in (\ref{Zb1b2-fin})}
\label{fig5}
\end{figure}

\section{$q$-difference equations for generating functions}

The foregoing expression (\ref{Zb1b2-fin}) 
of the open string amplitudes can be used 
to derive $q$-difference equation for the generating functions 
\begin{align}
  \Psi(x) &= \frac{1}{Z^{\ctv}_{\emptyset\emptyset}}
             \sum_{k=0}^\infty Z^{\ctv}_{(1^k)\emptyset}x^k,\\
  \tilde{\Psi}(x) &= \frac{1}{Z^{\ctv}_{\emptyset\emptyset}}
                     \sum_{k=0}^\infty Z^{\ctv}_{(k)\emptyset}x^k 
\end{align}
of special subsets of the normalized amplitudes 
$Z^{\ctv}_{\beta_1\beta_2}/Z^{\ctv}_{\emptyset\emptyset}$.  
Note that $(1^k)$ ($k$-copies of $1$) and $(k)$ 
represent Young diagrams with a single column or row.  
These generating functions are 
the Baker-Akhiezer functions\footnote{
Speaking more precisely, it is rather $\Psi(-x)$ 
and $\tilde{\Psi}(x)$ that literally correspond 
to the dual pair of Baker-Akhiezer functions.  
Because of this, the $q$-difference equations 
for $\Psi(x)$ and $\tilde{\Psi}(x)$ presented below 
are not fully symmetric. This is also the case 
for another pair $\Phi(x)$ and $\tilde{\Phi}(x)$ 
of generating functions introduced below.} 
of an integrable hierarchy, and $x$ amounts 
to the spectral variable therein \cite{Takasaki13bis}. 
One can derive $q$-difference equations for 
generating functions 
of $Z^{\ctv}_{\emptyset(1^k)}/Z^{\ctv}_{\emptyset\emptyset}$ 
and $Z^{\ctv}_{\emptyset(k)}/Z^{\ctv}_{\emptyset\emptyset}$ as well, 
though they become slightly more complicated 
because of the presence of the factor $q^{\kappa(\beta_2)/2}$.

\subsection{Derivation of $q$-difference equation}

A key towards the derivation of a $q$-difference equation 
is to compare $\Psi(x)$ and $\tilde{\Psi}(x)$ with 
another pair of generating functions 
\begin{align}
  \Phi(x) &= \frac{1}{Y_{\emptyset\emptyset}}
             \sum_{k=0}^\infty Y_{(1^k)\emptyset}x^k,\\
  \tilde{\Phi}(x) &= \frac{1}{Y_{\emptyset\emptyset}}
                     \sum_{k=0}^\infty Y_{(k)\emptyset}x^k 
\end{align}
obtained from the the main part 
\begin{align}
  Y_{\beta_1\beta_2} 
  &= \langle\tp{\beta}_1|
     \Gamma_{-}(q^{-\rho})\Gamma_{+}(q^{-\rho})(-Q_1)^{L_0}
     \Gamma'_{-}(q^{-\rho})\Gamma'_{+}(q^{-\rho})(-Q_3)^{L_0} \notag\\
  &\quad\mbox{}\times
     \Gamma_{-}(q^{-\rho})\Gamma_{+}(q^{-\rho})
     (-Q_2)^{L_0}\Gamma'_{-}(q^{-\rho})\Gamma'_{+}(q^{-\rho})
     |\tp{\beta}_2\rangle 
  \label{Yb1b2}
\end{align}
of the fermionic expression (\ref{Zb1b2-fin}) of $Z^{\ctv}_{\beta_1\beta_2}$. 

Let us first note the following relation between 
the coefficients of $\Psi(x)$ and $\Phi(x)$. 

\begin{lemma}
The coefficients of the expansion 
\beqnn
  \Psi(x) = \sum_{k=0}^\infty a_kx^k,\quad 
  \Phi(x) = \sum_{k=0}^\infty b_kx^k,\quad 
  a_0 = b_0 = 1, 
\eeqnn
are related as 
\beq
  a_k = b_k\prod_{i=1}^k(1 - Q_1Q_2q^{i-1})^{-1} 
  \quad \text{for $k\ge 1$}. 
\label{ak-bk-rel}
\eeq
\end{lemma}

\begin{proof}
$a_k/b_k$ is given by the ratio of the values 
of the prefactor in (\ref{Zb1b2-fin}) 
for $\beta_1 = (1^k),\,\beta_2 = \emptyset$ 
and $\beta_1 = \beta_2 = \emptyset$: 
\beqnn
\begin{aligned}
  \frac{a_k}{b_k} 
  &=\prod_{i,j=1}^\infty(1 - Q_1Q_2q^{-\beta_{1i}-\tp{\beta}_{2j}+i+j-1})^{-1}
    /\prod_{i,j=1}^\infty(1 -Q_1Q_2q^{i+j-1})^{-1}\\
  &= \prod_{i=1}^k\prod_{j=1}^\infty(1 - Q_1Q_2q^{i+j-2})^{-1}
     /\prod_{i=1}^k\prod_{j=1}^\infty(1 - Q_1Q_2q^{i+j-1})^{-1}\\
  &= \prod_{i=1}^k(1 - Q_1Q_2q^{i-1})^{-1}. 
\end{aligned}
\eeqnn
\end{proof}

The next step is to derive a $q$-difference equation for $\Phi(x)$. 

\begin{lemma}
$\Phi(x)$ can be expressed in the infinite-product form 
\beq
  \Phi(x) = \prod_{i=1}^\infty
      \frac{(1 - Q_1q^{i-1/2}x)(1 - Q_1Q_2Q_3q^{i-1/2}x)}
      {(1- q^{i-1/2}x)(1 - Q_1Q_3q^{i-1/2}x)}, 
  \label{Phi-infprod}
\eeq
and satisfies the $q$-difference equation 
\beq
  \Phi(qx) 
  = \frac{(1-q^{1/2}x)(1-Q_1Q_3q^{1/2}x)}
    {(1-Q_1q^{1/2}x)(1-Q_1Q_2Q_3q^{1/2}x)}\Phi(x).
  \label{Phi-qd-eq}
\eeq
\end{lemma}

\begin{remark}
The infinite product $\prod_{i=1}^\infty(1 - q^{i-1/2}x)^{-1}$ 
is an expression of the quantum dilogarithmic function \cite{FV93,FK93}. 
Thus $\Phi(x)$ is a multiplicative combination 
of four quantum dilogarithmic functions.  
\end{remark}

\begin{proof}
$Y_{(1^k)\emptyset}$ can be expressed as 
\beqnn
\begin{aligned}
  Y_{(1^k)\emptyset} 
  &= \langle(k)|\Gamma_{-}(q^{-\rho})\Gamma_{+}(q^{-\rho})(-Q_1)^{L_0}
     \Gamma'_{-}(q^{-\rho})\Gamma'_{+}(q^{-\rho})(-Q_3)^{L_0}\\
  &\quad\mbox{}\times
     \Gamma_{-}(q^{-\rho})\Gamma_{+}(q^{-\rho})
     (-Q_2)^{L_0}\Gamma'_{-}(q^{-\rho})\Gamma'_{+}(q^{-\rho})|0\rangle. 
\end{aligned}
\eeqnn
By the fundamental properties 
\beq
  \sum_{k=0}^\infty x^k\langle(k)| = \langle 0|\Gamma_{+}(x),\quad 
  \sum_{k=0}^\infty x^k\langle(1^k)| = \langle 0|\Gamma'_{+}(x), 
\eeq
of the single-variate vertex operators \cite{OR01,BY08}, 
the generating function of $Y_{(1^k)\emptyset}$'s 
can be expressed as 
\beqnn
\begin{aligned}
  \sum_{k=0}^\infty Y_{(1^k)\emptyset}x^k 
  &= \langle 0|\Gamma_{+}(x)
     \Gamma_{-}(q^{-\rho})\Gamma_{+}(q^{-\rho})(-Q_1)^{L_0}
     \Gamma'_{-}(q^{-\rho})\Gamma'_{+}(q^{-\rho})(-Q_3)^{L_0}\\
  &\quad\mbox{}\times
     \Gamma_{-}(q^{-\rho})\Gamma_{+}(q^{-\rho})
     (-Q_2)^{L_0}\Gamma'_{-}(q^{-\rho})\Gamma'_{+}(q^{-\rho})|0\rangle. 
\end{aligned}
\eeqnn
One can now use the commutation relations \cite{OR01,BY08} 
\begin{align}
  \Gamma_{+}(x)\Gamma_{-}(y) 
    &= (1 - xy)^{-1}\Gamma_{-}(y)\Gamma_{+}(x),\notag\\
  \Gamma'_{+}(x)\Gamma'_{-}(y) 
    &= (1 - xy)^{-1}\Gamma'_{-}(y)\Gamma'_{+}(x),\notag\\
  \Gamma_{+}(x)\Gamma'_{-}(y) 
    &= (1 + xy)\Gamma'_{-}(y)\Gamma_{+}(x),\notag\\
  \Gamma'_{+}(x)\Gamma_{-}(y)
    &= (1 + xy)\Gamma_{-}(y)\Gamma'_{+}(x) 
  \label{Gamma-com-rel}
\end{align}
of the single-variate vertex operators 
to move $\Gamma_{+}(x)$ to the right 
until it hits $|0\rangle$ and disappears.  
This yields the infinite-product expression 
\beqnn
  \sum_{k=0}^\infty Y_{(1^k)\emptyset}x^k  
  = \prod_{i=1}^\infty 
    \frac{(1-Q_1q^{i-1/2}x)(1-Q_1Q_2Q_3q^{i-1/2}x)}
    {(1-q^{i-1/2}x)(1-Q_1Q_3q^{i-1/2}x)}
    Y_{\emptyset\emptyset}
\eeqnn
of the unnormalized generating function, 
hence the expression (\ref{Phi-infprod}) of $\Phi(x)$.  
The $q$-difference equation (\ref{Phi-qd-eq}) 
is an immediate consequence of (\ref{Phi-infprod}). 
\end{proof}

To derive a $q$-difference equation for $\Psi(x)$, 
let us rewrite (\ref{Phi-qd-eq}) as 
\beqnn
\begin{aligned}
  &(1 - Q_1(1+Q_2Q_3)q^{1/2}x + Q_1^2Q_2Q_3qx^2)\Phi(qx)\\
  &= (1 - (1+Q_1Q_3)q^{1/2}x + Q_1Q_3qx^2)\Phi(x)
\end{aligned}
\eeqnn
and extract the coefficients of $x^k$. 
This yields the recursion relations 
\begin{align}
  &q^kb_k - Q_1(1+Q_2Q_3)q^{1/2}q^{k-1}b_{k-1} 
   + Q_1^2Q_2Q_3qq^{k-2}b_{k-2}\notag\\
  &= b_k - (1+Q_1Q_3)q^{1/2}b_{k-1} + Q_1Q_3qb_{k-2} 
  \label{bk-recursion}
\end{align}
for $b_k$'s.  Note that these relations hold for all $k\in\ZZ$ 
if $b_k$'s for $k<0$ are understood to be $0$. 
By (\ref{ak-bk-rel}), these recursion relations turn 
into the recursion relations 
\begin{align}
  &(1 - Q_1Q_2q^{k-2})(1 - Q_1Q_2q^{k-1})q^ka_k \notag\\
  &\mbox{} - Q_1(1+Q_2Q_3)q^{1/2}(1 - Q_1Q_2q^{k-2})q^{k-1}a_{k-1} 
    + Q_1^2Q_2Q_3qq^{k-2}a_{k-2} \notag\\
  &= (1 - Q_1Q_2q^{k-2})(1- Q_1Q_2q^{k-1})a_k \notag\\
  &\quad\mbox{} - (1+Q_1Q_3)q^{1/2}(1 - Q_1Q_2q^{k-2})a_{k-1} 
      + Q_1Q_3qa_{k-2} 
  \label{ak-recursion}
\end{align}
for $a_k$'s.  Multiplying these equations by $x^k$ 
and taking the sum over $k = 0,1,\ldots$, 
we can derive a $q$-difference equation for $\Phi(x)$.  

To state this result in a compact form, let us use 
the shift operator $q^{x\rd_x}$, $\rd_x = \rd/\rd x$, 
that acts on a function $f(x)$ of $x$ as 
\beqnn
  q^{x\rd_x}f(x) = f(qx). 
\eeqnn

\begin{theorem}\label{theorem2}
$\Psi(x)$ satisfies the $q$-difference equation 
\begin{align}
  &(1 - Q_1Q_2q^{-2}q^{x\rd_x})(1- Q_1Q_2q^{-1}q^{x\rd_x})\Psi(qx)\notag\\
  &\mbox{} - Q_1(1+Q_2Q_3)q^{1/2}x(1 - Q_1Q_2q^{-1}q^{x\rd_x})\Psi(qx) 
    + Q_1^2Q_2Q_3qx^2\Psi(qx) \notag\\
  &= (1 - Q_1Q_2q^{-2}q^{x\rd_x})(1 - Q_1Q_2q^{-1}q^{x\rd_x})\Psi(x)\notag\\
  &\quad\mbox{} - (1+Q_1Q_3)q^{1/2}x(1 - Q_1Q_2q^{-1}q^{x\rd_x})\Psi(x) 
     + Q_1Q_3qx^2\Psi(x).
  \label{Psi-qd-eq}
\end{align}
\end{theorem}

A $q$-difference equation for $\tilde{\Psi}(x)$ can be derived 
in the same way from the $q$-difference equation 
\beq
  \tilde{\Phi}(qx) 
  = \frac{(1 + Q_1q^{1/2}x)(1 + Q_1Q_2Q_3q^{1/2}x)}
    {(1 + q^{1/2}x)(1 + Q_1Q_3q^{1/2}x)}\tilde{\Phi}(x)
  \label{tPhi-qd-eq}
\eeq
for $\tilde{\Phi}(x)$ and the relation 
\beq
  \tilde{a}_k = \tilde{b}_k\prod_{i=1}^k(1 - Q_1Q_2q^{1-i})^{-1} 
  \quad \text{for $k \ge 1$} 
  \label{tak-tbk-rel}
\eeq
between the coefficients of the expansion 
\beqnn
  \tilde{\Psi}(x) = \sum_{k=0}^\infty \tilde{a}_kx^k,\quad 
  \tilde{\Phi}(x) = \sum_{k=0}^\infty \tilde{b}_kx^k,\quad 
  \tilde{a}_0 = \tilde{b}_0 = 1, 
\eeqnn
of $\tilde{\Psi}(x)$ and $\tilde{\Phi}(x)$. 
We omit the detail of calculation and show the result: 

\begin{theorem}\label{theorem3}
$\tilde{\Psi}(x)$ satisfies the $q$-difference equation 
\begin{align}
  &(1 - Q_1Q_2q^2q^{-x\rd_x})(1 - Q_1Q_2qq^{-x\rd_x})\tilde{\Psi}(q^{-1}x)\notag\\
  &\mbox{} + Q_1(1+Q_2Q_3)q^{-1/2}x(1 - Q_1Q_2qq^{-x\rd_x})\tilde{\Psi}(q^{-1}x) 
    + Q_1^2Q_2Q_3q^{-1}x^2\tilde{\Psi}(q^{-1}x) \notag\\
  &= (1 - Q_1Q_2q^2q^{-x\rd_x})(1 - Q_1Q_2qq^{-x\rd_x})\tilde{\Psi}(x)\notag\\
  &\quad\mbox{} + (1+Q_1Q_3)q^{-1/2}x(1 - Q_1Q_2qq^{-x\rd_x})\tilde{\Psi}(x) 
    + Q_1Q_3q^{-1}x^2\tilde{\Psi}(x). 
\label{tPsi-qd-eq}
\end{align}
\end{theorem}

\begin{remark}
Both sides of (\ref{Psi-qd-eq}) and (\ref{tPsi-qd-eq}) 
can be rewritten as 
\begin{align}
  &(1 - Q_1Q_2q^{-2}q^{x\rd_x} - Q_1q^{1/2}x) 
   (1 - Q_1Q_2q^{-1}q^{x\rd_x} - Q_1Q_2Q_3q^{1/2}x)\Psi(qx)\notag\\
  &= (1 - Q_1Q_2q^{-2}q^{x\rd_x} - q^{1/2}x) 
     (1 - Q_1Q_2q^{-1}q^{x\rd_x} - Q_1Q_3q^{1/2}x)\Psi(x) 
\end{align}
and 
\begin{align}
  &(1 - Q_1Q_2q^2q^{-x\rd_x} + Q_1q^{-1/2}x)
   (1 - Q_1Q_2qq^{-x\rd_x} + Q_1Q_2Q_3q^{-1/2}x)\tilde{\Psi}(q^{-1}x)\notag\\
  &= (1 - Q_1Q_2q^2q^{-x\rd_x} + q^{-1/2}x) 
     (1 - Q_1Q_2qq^{-x\rd_x} + Q_1Q_3q^{-1/2}x)\tilde{\Psi}(x). 
\end{align}
This expression corresponds to writing $q$-difference equations 
for $\Phi(x)$ and $\tilde{\Phi}(x)$ as 
\beq
  (1 - Q_1q^{1/2}x)(1 - Q_1Q_2Q_3q^{1/2}x)\Phi(qx)
  = (1 - q^{1/2}x)(1 - Q_1Q_3q^{1/2}x)\Phi(x)
\eeq
and 
\beq
  (1 + Q_1q^{-1/2}x)(1 + Q_1Q_2Q_3q^{-1/2}x)\tilde{\Phi}(q^{-1}x) 
  = (1 + q^{-1/2}x)(1 + Q_1Q_3q^{-1/2}x)\tilde{\Phi}(x). 
\eeq
These $q$-difference equations are transformed 
to the foregoing ones for $\Psi(x)$ and $\tilde{\Psi}(x)$ 
by the transformation (\ref{ak-bk-rel}) and (\ref{tak-tbk-rel}) 
of the coefficients.  

\end{remark}

\subsection{Structure of $q$-difference operators} 

Let us rewrite (\ref{Psi-qd-eq}) as 
\beq
  H(x,q^{x\rd_x})\Psi(x) = 0 
  \label{HPsi=0}
\eeq
and examine the structure of the $q$-difference operator 
$H$. This operator reads 
\begin{align}
  H(x,q^{x\rd_x}) 
  &= (1 - Q_1Q_2q^{-2}q^{x\rd_x})(1 - Q_1Q_2q^{-1}q^{x\rd_x})\notag\\
  &\quad\mbox{} 
     - (1+Q_1Q_3)q^{1/2}x(1 - Q_1Q_2q^{-1}q^{x\rd_x}) 
     + Q_1Q_3qx^2 \notag\\
  &\quad\mbox{} 
     - (1 - Q_1Q_2q^{-2}q^{x\rd_x})(1- Q_1Q_2q^{-1}q^{x\rd_x})q^{x\rd_x}\notag\\
  &\quad \mbox{} 
     + Q_1(1+Q_2Q_3)q^{1/2}x(1 - Q_1Q_2q^{-1}q^{x\rd_x})q^{x\rd_x}
     - Q_1^2Q_2Q_3qx^2q^{x\rd_x}. 
  \label{H}
\end{align}
Remarkably, $H(x,q^{x\rd_x})$ can be factorized as 
\beq
  H(x,q^{x\rd_x}) 
  = (1 - Q_1Q_2q^{-2}q^{x\rd_x})K(x,q^{x\rd_x}), 
\eeq
where 
\begin{align}
  K(x,q^{x\rd_x}) 
  &= (1 - Q_1Q_2q^{-1}q^{x\rd_x})(1 - q^{x\rd_x}) - (1+Q_1Q_3)q^{1/2}x \notag\\
  &\quad\mbox{} 
     + Q_1(1+Q_2Q_3)q^{1/2}xq^{x\rd_x} + Q_1Q_3qx^2. 
  \label{K}
\end{align}

This is also the case for the $q$-difference equation 
(\ref{tPsi-qd-eq}) for $\tilde{\Psi}(x)$.  
The $q$-difference operator $\tilde{H}(x,q^{x\rd_x})$ 
in the expression 
\beq
  \tilde{H}(x,q^{x\rd_x})\tilde{\Psi}(x) = 0 
  \label{tHtPsi=0}
\eeq
of (\ref{tPsi-qd-eq}) reads 
\begin{align}
  \tilde{H}(x,q^{x\rd_x}) 
  &= (1 - Q_1Q_2q^2q^{-x\rd_x})(1 - Q_1Q_2qq^{-x\rd_x})\notag\\
  &\quad\mbox{} + (1+Q_1Q_3)q^{1/2}x(1 - Q_1Q_2qq^{-x\rd_x})  
    + Q_1Q_3qx^2 \notag\\
  &\quad\mbox{}
    - (1 - Q_1Q_2q^2q^{-x\rd_x})(1 - Q_1Q_2qq^{-x\rd_x})q^{-x\rd_x}\notag\\
  &\quad\mbox{} 
    - Q_1(1+Q_2Q_3)q^{1/2}x(1 - Q_1Q_2qq^{-x\rd_x})q^{-x\rd_x} 
    - Q_1^2Q_2Q_3qx^2q^{-x\rd_x}.
  \label{tH}
\end{align}
This operator can be factorized as 
\beq
  \tilde{H}(x,q^{x\rd_x}) 
  = (1 - Q_1Q_2q^2q^{-x\rd_x})\tilde{K}(x,q^{x\rd_x}), 
\eeq
where 
\begin{align}
  \tilde{K}(x,q^{x\rd_x}) 
  &= (1 - Q_1Q_2qq^{-x\rd_x})(1 - q^{-x\rd_x}) + (1+Q_1Q_3)q^{1/2}x \notag\\
  &\quad\mbox{} 
     - Q_1(1+Q_2Q_3)q^{1/2}xq^{-x\rd_x} + Q_1Q_3qx^2. 
  \label{tK}
\end{align}

Let us note here that the action of 
$1 - Q_1Q_2q^{-2}q^{x\rd_x}$ and $1 - Q_1Q_2q^2q^{-x\rd_x}$ 
on the space of power series of $x$ is invertible 
as far as $Q_1$ and $Q_2$ take {\it generic values\/}, 
i.e., apart from the exceptional cases 
where $Q_1Q_2 = q^n$, $n \in \ZZ$.  
Therefore these factors can be removed 
from the $q$-difference equations (\ref{HPsi=0}) 
and (\ref{tHtPsi=0}). Actually, 
this genericity is implicitly assumed in the transformations 
(\ref{ak-bk-rel}) and (\ref{tak-tbk-rel}) 
of these generating functions.  Thus we find 
the following refinement of Theorems 
\ref{theorem2} and \ref{theorem3}. 

\begin{theorem}
For generic values of $Q_1$ and $Q_2$, 
the $q$-difference equations (\ref{Psi-qd-eq}) 
and (\ref{tPsi-qd-eq}) can be reduced to 
\beq
  K(x,q^{x\rd_x})\Psi(x) = 0,\quad 
  \tilde{K}(x,q^{x\rd_x})\tilde{\Psi}(x) = 0. 
  \label{red-qd-eq}
\eeq
\end{theorem}

This result fits well into the perspectives 
of mirror geometry of topological string theory 
on non-compact toric Calabi-Yau threefolds \cite{ADKMV03,DV07}.  
$\Psi(x)$ and $\tilde{\Psi}(x)$ may be thought of 
as {\it wave functions} of a probe D-brane.  
In this interpretation, a $q$-difference equation 
satisfied by these functions defines 
a {\it quantum mirror curve}.  
The $q$-difference equations (\ref{red-qd-eq}) 
indeed have such a characteristic.  
In the classical limit as $q\to1$, 
the non-commutative polynomials $K(x,q^{x\rd_x})$ 
and $\tilde{K}(x,q^{-x\rd_x})$ turn into the ordinary polynomials 
\begin{align}
  K_{\cl}(x,y) &= (1 - Q_1Q_2y)(1 - y) - (1+Q_1Q_3)x \notag\\
  &\quad\mbox{} + Q_1(1+Q_2Q_3)xy + Q_1Q_3x^2
\end{align}
in $(x,y)$ and 
\begin{align}
   \tilde{K}_{\cl}(x,y) &= (1 - Q_1Q_2y^{-1})(1 - y^{-1}) + (1+Q_1Q_3)x \notag\\
   &\quad\mbox{} - Q_1(1+Q_2Q_3)xy^{-1} + Q_1Q_3x^2 
\end{align}
in $(x,y^{-1})$.  As expected from the perspectives 
of mirror geometry, the Newton polygons of these polynomials 
have the same shape as the toric diagram in Figure \ref{fig1}.

\section{Flop transition}

Let us examine the flop transition from Figure \ref{fig1} 
to Figure \ref{fig6}.  After this move, the previous setup 
for defining the amplitude $Z^{\ctv}_{\beta_1\beta_2}$ 
turns into the setup shown in Figure \ref{fig7}.  
Note that the K\"ahler parameters after the flop transition 
are denoted by $P_1,P_2,P_3$; they are expected 
to be related to the K\"ahler parameters $Q_1,Q_2,Q_3$ 
before the transition by birational transformations. 

\begin{figure}
\centering
\includegraphics[scale=0.7]{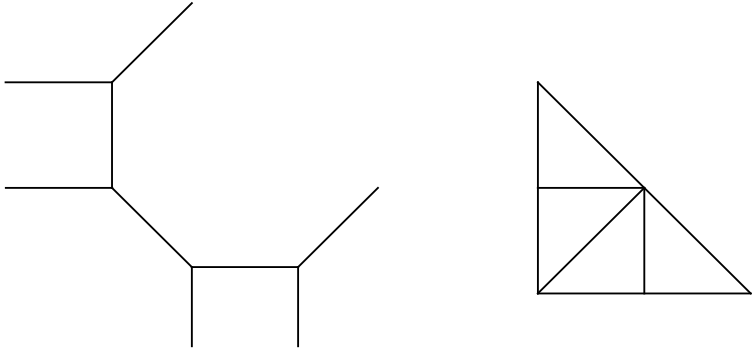}
\caption{Web and toric diagrams after flop}
\label{fig6}
\end{figure}

\begin{figure}
\centering
\includegraphics[scale=0.8]{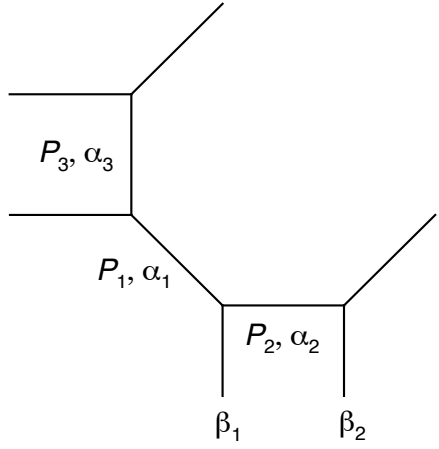}
\caption{Setup for open string amplitude $\hat{Z}^{\ctv}_{\beta_1\beta_2}$}
\label{fig7}
\end{figure}

Our method for calculating $Z^{\ctv}_{\beta_1\beta_2}$ 
can be extended to the amplitude $\hat{Z}^{\ctv}_{\beta_1\beta_2}$ 
of Figure \ref{fig7} as follows.  

The sum over $\alpha_1,\alpha_2,\alpha_3\in\calP$ can be 
decomposed to a partial sum $\hat{Z}^{\alpha_3}_{\beta_1\beta_2}$ 
with respect to $\alpha_1,\alpha_2$ at the first stage 
and a sum with respect to $\alpha_3$ at the next stage as 
\beq
  \hat{Z}^{\ctv}_{\beta_1\beta_2} 
  = \sum_{\alpha_3\in\calP}\hat{Z}_{\beta_1\beta_2|\alpha_3}
    (-P_3)^{|\alpha_3|}(-1)^{|\alpha_3|}q^{-\kappa(\alpha_3)/2}
    C_{\tp{\alpha}_3\emptyset\emptyset}. 
  \label{hZb1b2-fac}
\eeq
The extra factor $(-1)^{|\alpha_3|}q^{-\kappa(\alpha_3)/2}$ is inserted 
by the gluing rule.  The framing number (\ref{f-number}) 
along the internal line carrying $\alpha_3$ is equal to $1$.  

The partial sum $\hat{Z}_{\beta_1\beta_2|\alpha_3}$ 
is an open string amplitude of the double-$\PP^1$ diagram 
shown in Figure \ref{fig8}.  Since this is an on-strip diagram, 
the amplitude can be calculated explicitly as 
\begin{align}
  \hat{Z}_{\beta_1\beta_2|\alpha_3} 
  &= s_{\tp{\beta}_1}(q^{-\rho})s_{\tp{\beta}_2}(q^{-\rho})
     s_{\tp{\alpha}_3}(q^{-\rho}) \prod_{i,j=1}^\infty 
      (1 - P_2q^{-\beta_{1i}-\tp{\beta}_{2j}+i+j-1})^{-1} \notag\\
  &\quad\mbox{} \times 
      \prod_{i,j=1}^\infty(1 - P_1q^{-\tp{\alpha}_{3i}-\tp{\beta}_{1j}+i+j-1}) 
      \prod_{i,j=1}^\infty(1 - P_1P_2q^{-\tp{\alpha}_{3i}-\tp{\beta}_{2j}+i+j-1}). 
  \label{hZb1b2a3}
\end{align}
This amplitude is related to its counterpart $Z_{\beta_1\beta_2|\alpha_3}$ 
by the same flop operation as the move 
from Figure \ref{fig1} to Figure \ref{fig6}. 
One see form (\ref{Zb1b2a3}) and (\ref{hZb1b2a3}) 
that $\hat{Z}_{\beta_1\beta_2|\alpha_3}$ is almost identical 
to $Z_{\beta_1\beta_2|\alpha_3}$ if the K\"ahler parameters 
are related as 
\beq
  P_1P_2 = Q_2,\quad P_2 = Q_1Q_2. 
  \label{PQ12-rel}
\eeq
The only discrepancy lies in the infinite products 
$\prod_{i,j=1}^\infty(1 - Q_1q^{\cdots})$ in (\ref{Zb1b2a3}) 
and $\prod_{i,j=1}^\infty(1 - P_1q^{\cdots})$ in (\ref{hZb1b2a3}). 

\begin{figure}
\centering
\includegraphics[scale=0.8]{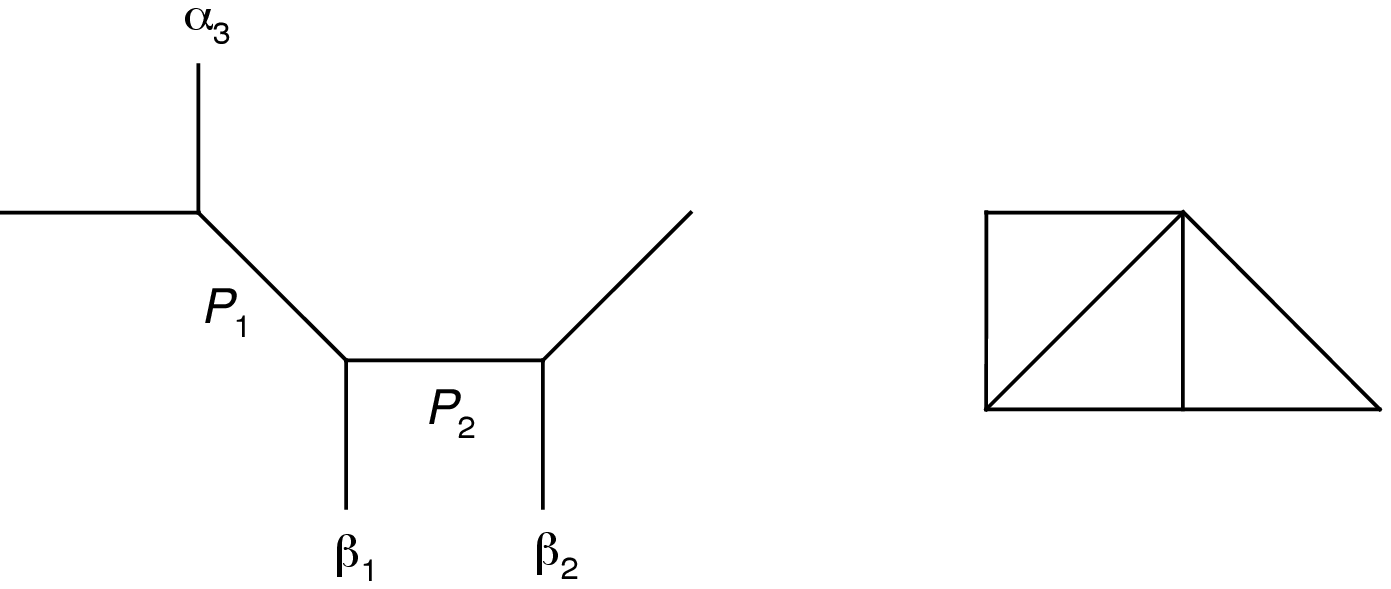}
\caption{Double-$\PP^1$ diagram defining $\hat{Z}_{\beta_1\beta_2|\alpha_3}$}
\label{fig8}
\end{figure}

Substituting (\ref{hZb1b2a3}) and (\ref{Ca3}) 
in (\ref{hZb1b2-fac}), we obtain the following expression 
of $\hat{Z}^{\ctv}_{\beta_1\beta_2}$: 
\begin{align}
  \hat{Z}^{\ctv}_{\beta_1\beta_2} 
  &= s_{\tp{\beta}_1}(q^{-\rho})s_{\tp{\beta}_2}(q^{-\rho})
     \prod_{i,j=1}^\infty(1 - P_2q^{-\beta_{1i}-\tp{\beta}_{2j}+i+j-1})^{-1}
     \notag\\
  &\quad\mbox{}\times \sum_{\alpha_3\in\calP}
     s_{\alpha_3}(q^{-\rho})^2P_3^{|\alpha_3|}  
     \prod_{i,j=1}^\infty(1 - P_1q^{-\tp{\alpha}_{3i}-\tp{\beta}_{1j}+i+j-1})
     \notag\\
  &\quad\quad\quad\mbox{}\times 
     \prod_{i,j=1}^\infty(1 - P_1P_2q^{-\tp{\alpha}_{3i}-\tp{\beta}_{2j}+i+j-1}). 
  \label{hZb1b2-red}
\end{align}
Note that we have used the identity (\ref{tp-rel}) as well 
to rewrite the first part of the summand as 
\beqnn
  s_{\tp{\alpha}_3}(q^{-\rho})s_{\alpha_3}(q^{-\rho})q^{-\kappa(\alpha_3)/2}
  = s_{\alpha_3}(q^{-\rho})^2. 
\eeqnn 
Thus, in contrast with (\ref{Zb1b2-red}), 
the sum in this case resembles the partition function 
of the ordinary melting crystal model \cite{NT07,NT08} 
for which the main part of the Boltzmann weight 
is $s_{\alpha_3}(q^{-\rho})^2$ 
rather than $s_{\tp{\alpha}_3}(q^{-\rho})s_{\alpha_3}(q^{-\rho})$. 

The sum in (\ref{hZb1b2-red}) can be calculated 
in more or less the same way as the case of (\ref{Zb1b2-red}).  
Let us show the final result only. 

\begin{theorem}
The open string amplitude $\hat{Z}^{\ctv}_{\beta_1\beta_2}$ 
can be expressed as 
\begin{align}
  \hat{Z}^{\ctv}_{\beta_1\beta_2} 
  &= q^{\kappa(\beta_1)/2+\kappa(\beta_2)/2}
     \prod_{i,j=1}^\infty(1 - P_2q^{-\beta_{1i}-\tp{\beta}_{2j}+i+j-1})^{-1}
     \notag\\
  &\quad\mbox{}\times 
     \langle\tp{\beta}_1| 
     \Gamma'_{-}(q^{-\rho})\Gamma'_{+}(q^{-\rho})(-P_1)^{L_0}
     \Gamma_{-}(q^{-\rho})\Gamma_{+}(q^{-\rho})P_3^{L_0} \notag\\
  &\quad\mbox{}\times 
     \Gamma_{-}(q^{-\rho})\Gamma_{+}(q^{-\rho})(-P_1P_2)^{L_0}
     \Gamma'_{-}(q^{-\rho})\Gamma'_{+}(q^{-\rho})
     |\tp{\beta}_2\rangle. 
  \label{hZb1b2-fin}
\end{align}
\end{theorem}

The main part $\langle\tp{\beta}_1|\cdots|\tp{\beta}_2\rangle$ 
of this expression is essentially the open string amplitude 
of the web diagram shown in Figure \ref{fig9}.  
This web diagram can be derived the web diagram 
of Figure \ref{fig5} by the same flop operation 
as the move from Figure \ref{fig1} to Figure \ref{fig6}. 

\begin{figure}
\centering
\includegraphics[scale=0.8]{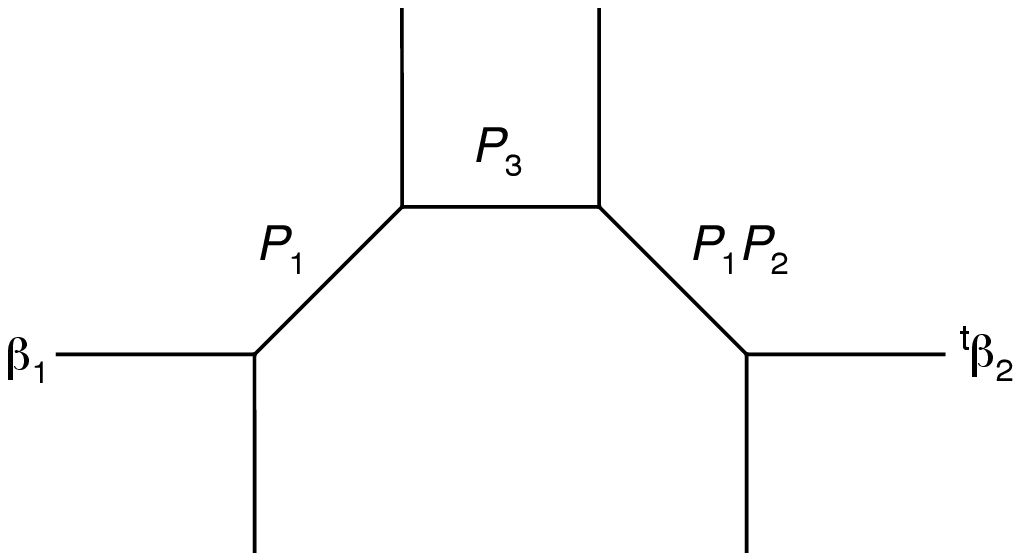}
\caption{Web diagram emerging in (\ref{hZb1b2-fin})} 
\label{fig9}
\end{figure}

To see how this part is related to the main part 
of (\ref{Zb1b2-fin}), let us use the commutation relations 
(\ref{Gamma-com-rel}) to exchange the order of 
the first four vertex operators therein as 
\beqnn
\begin{aligned}
  &\langle\tp{\beta}_1| 
     \Gamma'_{-}(q^{-\rho})\Gamma'_{+}(q^{-\rho})(-P_1)^{L_0}
     \Gamma_{-}(q^{-\rho})\Gamma_{+}(q^{-\rho})P_3^{L_0} \\
  &= \langle\tp{\beta}_1| (-P_1)^{L_0} 
     \Gamma'_{-}(-P_1^{-1}q^{-\rho})\Gamma'_{+}(-P_1q^{-\rho})
     \Gamma_{-}(q^{-\rho})\Gamma_{+}(q^{-\rho})P_3^{L_0} \\
  &= (-P_1)^{|\beta_1|}
     \prod_{i,j=1}^\infty(1 - P_1q^{i+j-1})(1 - P_1^{-1}q^{i+j-1})^{-1}\\
  &\quad\mbox{}\times
     \langle\tp{\beta}_1| 
     \Gamma_{-}(q^{-\rho})\Gamma_{+}(q^{-\rho})
     \Gamma'_{-}(-P_1^{-1}q^{-\rho})\Gamma'_{+}(-P_1q^{-\rho})P_3^{L_0} \\
  &= (-P_1)^{|\beta_1|}
     \prod_{i,j=1}^\infty(1 - P_1q^{i+j-1})(1 - P_1^{-1}q^{i+j-1})^{-1}\\
  &\quad\mbox{}\times 
     \langle\tp{\beta}_1| 
     \Gamma_{-}(q^{-\rho})\Gamma_{+}(q^{-\rho})(-P_1^{-1})^{L_0}
     \Gamma'_{-}(q^{-\rho})\Gamma'_{+}(q^{-\rho})(-P_1P_3)^{L_0}. 
\end{aligned}
\eeqnn
This shows that if the two sets of K\"ahler parameters 
are matched as 
\beq
  P_1^{-1} = Q_1, \quad P_1P_3 = Q_3,\quad P_1P_2 = Q_2, 
  \label{PQ123-rel}
\eeq
$\hat{Z}^{\ctv}_{\beta_1\beta_2}$ and $Z^{\ctv}_{\beta_1\beta_2}$ 
are related as 
\beq
  \hat{Z}^{\ctv}_{\beta_1\beta_2} 
  = q^{\kappa(\beta_1)/2}(-P_1)^{|\beta_1|}
    \prod_{i,j=1}^\infty(1 - P_1q^{i+j-1})(1 - P_1^{-1}q^{i+j-1})^{-1}
    \cdot Z^{\ctv}_{\beta_1\beta_2}. 
  \label{ZhZ-rel}
\eeq
Note that (\ref{PQ123-rel}) is consistent with (\ref{PQ12-rel}). 
These matching rules of parameters agree with the known result 
for the partition functions \cite{IKP04,Sulkowski06,KM06}. 

\begin{remark}\label{remark5}
It is instructive to examine a different 
cut-and-glue procedure in this case.  
Let us try to cut the middle internal line 
(to which $Q_1$ and $\alpha_1$ are assigned) 
of the web diagram (see Figure \ref{fig7}).  
The cutting procedure yields two subdiagrams 
of the on-strip type.  They are glued together 
with the edge weight $(-P_1)^{|\alpha_1|}$. 
Note that the framing number in this case 
is equal to $0$.  Thus the total amplitude 
can be expressed as 
\beq
  \hat{Z}^{\ctv}_{\beta_1\beta_2} 
  = \sum_{\alpha_1\in\calP}\hat{Z}'_{\alpha_1}(-P_1)^{|\alpha_1|}
    \hat{Z}''_{\alpha_1|\beta_1\beta_2},
  \label{hZb1b2-fac2}
\eeq
where $\hat{Z}'_{\alpha_1}$ and $\hat{Z}''_{\alpha_1|\beta_1\beta_2}$ 
are contributions of the two on-strip subdiagrams, i.e., 
\beqnn
  \hat{Z}'_{\alpha_1} 
  = \langle0|\Gamma'_{-}(q^{-\rho})\Gamma'_{+}(q^{-\rho}) 
    P_3^{L_0}\Gamma'_{-}(q^{-\rho})\Gamma'_{+}(q^{-\rho})
    |\alpha_1\rangle 
\eeqnn
and 
\begin{align*}
  \hat{Z}''_{\alpha_1|\beta_1\beta_2} 
  &= q^{-\kappa(\alpha_1)/2}
     s_{\tp{\beta}_1}(q^{-\rho})s_{\tp{\beta}_2}(q^{-\rho})\notag\\
  &\quad\mbox{}\times
     \langle\alpha_1|\Gamma'_{-}(q^{-\tp{\beta}_1-\rho})
     \Gamma'_{+}(q^{-\beta_1-\rho})P_2^{L_0}
     \Gamma'_{-}(q^{-\tp{\beta}_2-\rho})
     \Gamma'_{+}(q^{-\beta_2-\rho})|0\rangle. 
\end{align*}
Plugging these expressions into (\ref{hZb1b2-fac2}) 
leads to yet another fermionic expression 
of $\hat{Z}^{\ctv}_{\beta_1\beta_2}$: 
\begin{align}
  \hat{Z}^{\ctv}_{\beta_1\beta_2}
  &= s_{\tp{\beta}_1}(q^{-\rho})s_{\tp{\beta}_2}(q^{-\rho})\notag\\
  &\quad\mbox{}\times \langle0|\Gamma'_{-}(q^{-\rho})\Gamma'_{+}(q^{-\rho})
     P_3^{L_0}\Gamma'_{-}(q^{-\rho})\Gamma'_{+}(q^{-\rho})
     q^{-K/2}(-P_1)^{L_0}\notag\\
  &\quad\mbox{}\times 
     \Gamma'_{-}(q^{-\tp{\beta}_1-\rho})
     \Gamma'_{+}(q^{-\beta_1-\rho})P_2^{L_0}
     \Gamma'_{-}(q^{-\tp{\beta}_2-\rho})
     \Gamma'_{+}(q^{-\beta_2-\rho})|0\rangle. 
  \label{hZb1b2-bis}
\end{align}
This expression looks very similar to an on-strip amplitude.  
The operator product in this expression, however, contains 
the operator $q^{-K/2}$ that does not appear in on-strip amplitudes.  
Because of this operator, one cannot calculate this expression directly.  
In contrast, if one applies the cut-and-glue procedure 
to an on-strip amplitude, operators of the form $q^{\pm K/2}$ 
do not appear or cancel out in the outcome 
of calculation\footnote{This is a key to prove the fermionic formula 
(\ref{Zstrip-fermion}) of on-strip amplitudes by induction.}. 
This cancellation mechanism is a consequence 
of the {\it linear} shape of the on-strip diagram.  
In this respect, the web diagram of Figure \ref{fig6} 
is a {\it chain} of on-strip diagrams, and its web diagram 
is {\it bent} to ninety degrees in the middle.  
It is this bend that generates the operator $q^{-K/2}$.  
Actually, the present case is special in the sense that 
this difficulty can be circumvented by the foregoing 
different cut-and-glue description\footnote{
One can also convert (\ref{hZb1b2-bis}) to a more tractable form 
with the aid of techniques used in Sections 3 and 4.  
This eventually leads to the same result as presented therein.}. 
In a general case, such an escape route is not prepared. 
\end{remark}

\section{Conclusion}

Let us summarize what we have done in this paper.  

\paragraph{Calculation of  open string amplitudes}
We reformulated the open string amplitude $Z^{\ctv}_{\beta_1\beta_2}$ 
of Figure \ref{fig2} in the partially summed form (\ref{Zb1b2-fac}), 
and derived the reduced expression (\ref{Zb1b2-red}). 
The main part of (\ref{Zb1b2-red}) turns out 
to be similar to the partition function 
of the modified melting crystal model.  Firstly, 
the main part $s_{\tp{\alpha}_3}(q^{-\rho})s_{\alpha_3}(q^{-\rho})$ 
of the summand is exactly the same.  Secondly,  
the other part can be described by matrix elements 
of the {\it diagonal} operators $V^{(\pm)}_0$ 
in the quantum torus algebra. 
This is also a characteristic of the external potentials 
in the melting crystal models. We could thereby apply 
the method for the melting crystal models to derive 
the fermionic expression (\ref{Zb1b2-med}) 
of $Z^{\ctv}_{\beta_1\beta_2}$. This expression was 
further converted to the final expression 
(\ref{Zb1b2-fin}) of $Z^{\ctv}_{\beta_1\beta_2}$, 
which is a product of a simple prefactor and 
the open string amplitude $Y_{\beta_1\beta_2}$ 
of a new on-strip diagram.  

\paragraph{Derivation of $q$-difference equations} 
We derived $q$-difference equations for 
the generating functions $\Psi(x),\,\tilde{\Psi}(x)$ 
of the normalized amplitudes 
$Z^{\ctv}_{\beta_1\beta_2}/Z^{\ctv}_{\emptyset\emptyset}$ 
specialized to $\beta_1 = (1^k),\,(k)$, $k = 0,1,2,\ldots$, 
and $\beta_2 = \emptyset$.  The derivation makes 
full use of the factorized form of (\ref{Zb1b2-fin}).  
Namely, we first derived the $q$-difference equations 
(\ref{Phi-qd-eq}) and (\ref{tPhi-qd-eq}) 
for the generating functions $\Phi(x),\,\tilde{\Phi}(x)$ 
obtained from $Y_{\beta_1\beta_2}/Y_{\emptyset\emptyset}$. 
These equations are transformed 
to the $q$-difference equations (\ref{Psi-qd-eq}) 
and (\ref{tPsi-qd-eq}) for $\Psi(x),\,\tilde{\Psi}(x)$. 
This is the place where the prefactor of $Y_{\beta_1\beta_2}$ 
in (\ref{Zb1b2-fin}) plays a role.  
We examined the structure of these $q$-difference equations 
and found that they can be reduced to the simpler equations 
(\ref{red-qd-eq}).  It is these reduced equations 
that should be interpreted as the defining equation 
of a quantum mirror curve. 

\paragraph{Flop transition}
We considered the flop transition from Figure \ref{fig1} 
to Figure \ref{fig6}.  The open string amplitude 
$\hat{Z}^{\ctv}_{\beta_1\beta_2}$ after the transition 
can be calculated in much the same way 
as in the case of $Z^{\ctv}_{\beta_1\beta_2}$.  
We confirmed that $\hat{Z}^\ctv_{\beta_1\beta_2}$  
can be matched to the amplitude $Z^{\ctv}_{\beta_1\beta_2}$ 
by the birational transformations (\ref{PQ123-rel}) 
of the K\"ahler parameters. 
\bigskip

On the other hand, we have been unable to derive 
$q$-difference equations in other configurations 
of partitions on the external lines of the web diagram 
(except for those that can be derived from the setup 
of Section 2 by symmetries or specializations of the amplitude).  
A major obstacle is the emergence of $q^{\pm K/2}$'s  
that do not cancel out in a fermionic expression 
of the amplitude as opposed to the case of (\ref{cancellation}). 
Because of this obstacle\footnote{The situation presented 
in Remark \ref{remark5} is similar, but the difficulty 
in that case can be circumvented.}, 
the fermionic expression in such a case cannot be converted 
to a form from which a $q$-difference equation can be read out.  

We have encountered the same difficulty in an attempt 
to extend our results to more general tree-like web diagrams 
studied by Karp, Liu and Mari\~no \cite{KLM05}.  
Our attempt has been unsuccessful not only 
for open string amplitudes, but also 
for the closed string partition function.  
We believe that this difficulty is of technical nature 
and can be overcome by a new computational idea.

\subsection*{Acknowledgements}

The authors are grateful to Motohico Mulase 
for valuable comments. 
This work is partly supported by JSPS Kakenhi Grant 
No. 24540223, No. 25400111 and No. 15K04912.

\appendix

\section{Amplitudes of on-strip geometry}

The toric diagram of on-strip geometry is a triangulation 
of the strip of height $1$ to triangles of area $1/2$ 
(see Figure \ref{fig10}).  
The associated web diagram is a connected acyclic graph.  
If the toric graph comprises $N$ triangles,  
the web diagram has $N$ vertices, $N-1$ internal lines 
and $N+2$ external lines.  The $N$ external lines 
other than the leftmost and rightmost ones are vertical.  
For brevity, the external lines are also referred to 
as ``legs''.  

\begin{figure}
\centering
\includegraphics[scale=0.9]{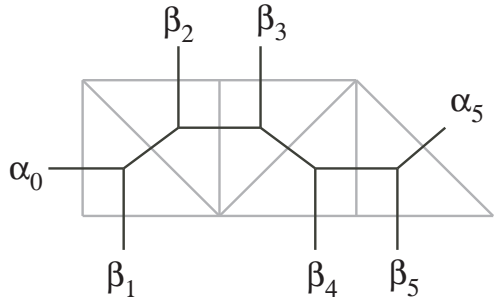}
\caption{Web diagram on triangulated strip} 
\label{fig10}
\end{figure}

We assign the K\"ahler parameters $Q_1,\ldots,Q_{N-1}$ 
to the internal lines, the partitions $\beta_1,\ldots,\beta_N$ 
to the vertical external lines, and the partitions $\alpha_0,\alpha_N$ 
to the leftmost and rightmost external lines.  
Let $Z^{\alpha_0\alpha_N}_{\beta_1\cdots\beta_N}$ denote 
the open string amplitude in this setup.  
This amplitude is defined as a sum of the product 
of vertex and edge weights with respect to 
the partitions $\alpha_1,\ldots,\alpha_{N-1}$ on the internal lines. 
In the case of $\alpha_0 = \alpha_N = \emptyset$, 
Iqbal and Kashani-Poor \cite{IKP04} calculated this sum 
in a closed form by skillful use of the Cauchy identities 
for skew Schur functions.  Their result can be reformulated, 
without restriction to $\alpha_0 = \alpha_N = \emptyset$, 
in the language of fermions \cite{EK03,Nagao09,Sulkowski09}. 

Following the notations of Nagao \cite{Nagao09} 
and Su{\l}kowski \cite{Sulkowski09}, let us define 
the sign (or type) $\sigma_n = \pm 1$ of the $n$-th vertex as: 
\begin{itemize}
\item[(i)] $\sigma_n = +1$ if the vertical leg points up, 
\item[(ii)] $\sigma_n = -1$ if the vertical leg points down. 
\end{itemize}
For example, in the case of the web diagram of Figure \ref{fig10}, 
\beqnn
  \sigma_1 = -1,\quad \sigma_2 = +1,\quad \sigma_3 = +1,\quad 
  \sigma_4 = -1,\quad \sigma_5 = -1. 
\eeqnn
These data are used to show the types of vertex operators as 
\beqnn
  \Gamma^\sigma_{\pm}(\bsx) 
  = \begin{cases}
    \Gamma_{\pm}(\bsx) &\text{if $\sigma = +1$},\\
    \Gamma'_{\pm}(\bsx) &\text{if $\sigma = -1$}. 
    \end{cases}
\eeqnn
Let us further introduce the auxiliary notations 
\beqnn
  \beta^{(n)} 
  = \begin{cases}
    \beta_n &\text{if $\sigma_n = +1$},\\
    \tp{\beta}_n &\text{if $\sigma_n = -1$}, 
    \end{cases}
  \qquad 
  Q_{mn} = Q_mQ_{m+1}\cdots Q_{n-1}. 
\eeqnn

With these notations, the fermionic expression 
of $Z^{\alpha_0\alpha_N}_{\beta_1\cdots\beta_N}$ read 
\begin{align}
  Z^{\alpha_0\alpha_N}_{\beta_1\cdots\beta_N} 
  &= q^{(1-\sigma_1)\kappa(\alpha_0)/4}
     q^{(1+\sigma_N)\kappa(\alpha_N)/4}
     s_{\tp{\beta}_1}(q^{-\rho})\cdots s_{\tp{\beta}_N}(q^{-\rho}) \notag\\
  &\quad\mbox{}\times 
     \langle\tp{\alpha}_0| \Gamma^{\sigma_1}_{-}(q^{-\beta^{(1)}-\rho})
     \Gamma^{\sigma_1}_{+}(q^{-\tp{\beta}^{(1)}-\rho})(\sigma_1Q_1\sigma_2)^{L_0}
     \cdots \notag\\
  &\quad\mbox{}\times  
     \Gamma^{\sigma_{N-1}}_{-}(q^{-\beta^{(N-1)}-\rho})
     \Gamma^{\sigma_{N-1}}_{+}(q^{-\tp{\beta}^{(N-1)}-\rho})
     (\sigma_{N-1}Q_{N-1}\sigma_N)^{L_0} \notag\\
  &\quad\mbox{}\times
     \Gamma^{\sigma_N}_{-}(q^{-\beta^{(N)}-\rho})
     \Gamma^{\sigma_N}_{+}(q^{-\tp{\beta}^{(N)}-\rho})|\alpha_N\rangle.
  \label{Zstrip-fermion}
\end{align}
In the case where $N = 1$, this formula reduces 
to the fermionic expression (\ref{Cabc-fermion}) 
of the vertex weight itself.  Starting from (\ref{Cabc-fermion}), 
one can prove this formula by induction.  
If $\alpha_0 = \alpha_N = \emptyset$, 
one can use the commutation relations (\ref{Gamma-com-rel}) 
to move $\Gamma^\sigma_{-}$'s to the left and 
$\Gamma^\sigma_{+}$'s to the right until 
they hit $\langle 0|$ and $|0\rangle$ and disappear.  
This yields the explicit formula 
\begin{align}
  Z^{\emptyset\emptyset}_{\beta_1\cdots\beta_N} 
  &= s_{\tp{\beta}_1}(q^{-\rho})\cdots s_{\tp{\beta}_N}(q^{-\rho}) \notag\\
  &\mbox{}\times \prod_{1\leq m<n\leq N}\prod_{i,j=1}^\infty 
     (1 - Q_{mn}q^{-\tp{\beta}^{(m)}_i-\beta^{(n)}_j+i+j-1})^{-\sigma_m\sigma_n} 
\end{align}
of Iqbal and Kashani-Poor \cite{IKP04}.

\section{Direct proof of two-leg cyclic symmetry} 

As another application of the techniques used 
in Section 3, we present a direct proof of the identities 
(\ref{2leg-CS4}) and (\ref{2leg-CS5}) that amounts 
to the cyclic symmetry of two-leg vertices.  
Actually, these two identities are equivalent, 
and can be reduced to the following one: 
\beq
  s_\lambda(q^{-\rho})s_\mu(q^{-\lambda-\rho}) 
  = \langle\mu|q^{-K/2}\Gamma'_{-}(q^{-\rho})
    \Gamma'_{+}(q^{-\rho})q^{-K/2}|\lambda\rangle. 
  \label{2leg-CSfin}
\eeq
It is this identity that we prove here.  
Note that this identity implies the non-trivial relation 
\beq
  s_\lambda(q^{-\rho})s_\mu(q^{-\lambda-\rho}) 
  = s_\mu(q^{-\rho})s_\lambda(q^{-\mu-\rho}), 
\eeq
from which the equivalence of (\ref{2leg-CS4}) 
and (\ref{2leg-CS5}) follows.  

We prove (\ref{2leg-CSfin}) by generating functions.  
Namely, we construct generating functions of both sides 
by the Schur functions $s_\mu(\bsx)$, $\bsx = (x_1,x_2,\ldots)$, 
and confirm that these generating functions are identical.  

It is easy to calculate the generating function 
of the left side of (\ref{2leg-CSfin}).  
By the Cauchy identity 
\beq
  \sum_{\mu\in\calP}s_\mu(\bsx)s_\mu(\bsy) 
  = \prod_{i,j=1}^\infty(1 - x_iy_j)^{-1}, 
  \quad \bsy = (y_1,y_2,\ldots), 
\eeq
of the Schur functions \cite{Mac-book}, 
the generating function of the left side 
of (\ref{2leg-CSfin}) can be expressed as 
\beq
  \sum_{\mu\in\calP}s_\mu(\bsx)
  s_\lambda(q^{-\rho})s_\mu(q^{-\lambda-\rho})
  = s_\lambda(q^{-\rho})\prod_{i,j=1}^\infty
      (1 - x_iq^{-\lambda_j+j-1/2})^{-1}. 
  \label{CSfin-lhs}
\eeq

On the other hand, constructing the generating function 
of the left side of (\ref{2leg-CSfin}) amounts to 
inserting $\Gamma_{+}(\bsx)$ to the right of $\langle 0|$ as 
\beqnn
\begin{aligned}
  &\sum_{\mu\in\calP}s_\mu(\bsx)\langle\mu|q^{-K/2}\Gamma'_{-}(q^{-\rho})
   \Gamma'_{+}(q^{-\rho})q^{-K/2}|\lambda\rangle\\
  &= \langle 0|\Gamma_{+}(\bsx)q^{-K/2}\Gamma'_{-}(q^{-\rho})
     \Gamma'_{+}(q^{-\rho})q^{-K/2}|\lambda\rangle\\
  &= \langle 0|\exp\left(\sum_{i,k=1}^\infty\frac{x_i^k}{k}J_k\right)
     q^{-K/2}\Gamma'_{-}(q^{-\rho})\Gamma'_{+}(q^{-\rho})q^{-K/2}
     |\lambda\rangle. 
\end{aligned}
\eeqnn
The subsequent calculation is very similar to Section 3.  
One can use (\ref{SS1}) and (\ref{SS3}) to rewrite 
the last quantity as 
\beqnn
\begin{aligned}
  &\langle 0|\exp\left(\sum_{i,k=1}^\infty\frac{x_i^k}{k}J_k\right)
     q^{-K/2}\Gamma'_{-}(q^{-\rho})\Gamma'_{+}(q^{-\rho})
     q^{-K/2}|\lambda\rangle\\
  &= \langle 0|q^{-K/2}\Gamma'_{-}(q^{-\rho})\Gamma'_{+}(q^{-\rho}) 
     \exp\left(\sum_{i,k=1}^\infty\frac{x_i^kq^{k/2}}{k}
       \left(V^{(-k)}_0 + \frac{1}{1-q^k}\right)
     \right) q^{-K/2}|\lambda\rangle\\
  &= \langle 0|\Gamma'_{+}(q^{-\rho})q^{-K/2}
     \exp\left(\sum_{i,k=1}^\infty\frac{x_i^kq^{k/2}}{k}
       \left(V^{(-k)}_0 + \frac{1}{1-q^k}\right)
     \right) |\lambda\rangle. 
\end{aligned}
\eeqnn
Note that the order of $\exp(\cdots)$ and $q^{K/2}$ 
has been exchanged because $V^{(-k)}_0$ commutes with $q^{K/2}$. 
By (\ref{V(-k)0-action}), the action of $\exp(\cdots)$ 
on $|\lambda\rangle$ can be expressed as  
\beqnn
\begin{aligned}
  &\exp\left(\sum_{i,k=1}^\infty\frac{x_i^kq^{k/2}}{k}
     \left(V^{(-k)}_0 + \frac{1}{1-q^k}\right)\right)|\lambda\rangle \\
  &= \exp\left(\sum_{i,k=1}^\infty\frac{x_i^kq^{k/2}}{k}
       \sum_{j=1}^\infty q^{-k(\lambda_j-j+1)}\right)|\lambda\rangle \\
  &= \prod_{i,j=1}^\infty \exp\left( \sum_{k=1}^\infty
       \frac{(x_iq^{-\lambda_j+j-1/2})^k}{k}\right) |\lambda\rangle \\
  &= \prod_{i,j=1}^\infty(1 - x_iq^{-\lambda_j+j-1/2})^{-1} 
     |\lambda\rangle. 
\end{aligned}
\eeqnn
Thus the generating function of the right side of (\ref{2leg-CSfin}) 
turns out to take such a form as 
\beqnn
\begin{aligned}
  &\sum_{\mu\in\calP}s_\mu(\bsx)\langle\mu|q^{-K/2}\Gamma'_{-}(q^{-\rho})
   \Gamma'_{+}(q^{-\rho})q^{-K/2}|\lambda\rangle \\
  &= \langle 0|\Gamma'_{+}(q^{-\rho})q^{-K/2}|\lambda\rangle 
     \prod_{i,j=1}^\infty(1 - x_iq^{-\lambda_j+j-1/2})^{-1} \\
  &= q^{-\kappa(\lambda)/2}s_{\tp{\lambda}}(q^{-\rho})
     \prod_{i,j=1}^\infty(1 - x_iq^{-\lambda_j+j-1/2})^{-1}. 
\end{aligned}
\eeqnn
By (\ref{tp-rel}), this coincides with (\ref{CSfin-lhs}).

\end{document}